%% file: Article.tex
\documentclass[a4paper,UKenglish,cleveref, autoref,numberwithinsect]{lipics-v2021}

\usepackage{tikz,tikz-qtree}
\usetikzlibrary{arrows,positioning,calc}

\usepackage{url}
\usepackage{wrapfig}
\usepackage{hhline}
\usepackage{multirow}
\usepackage{paralist}

\usepackage{multicol}
\usepackage{mathtools}

\usepackage{amsmath,amsfonts,amssymb}

\input{Macro}

\title{A quantitative extension of Interval Temporal Logic over infinite words}
\titlerunning{A quantitative extension of Interval Temporal Logic over infinite words}

\author{Laura Bozzelli}{University of Napoli ``Federico II'', Napoli, Italy}{}{}{}{}
\author{Adriano Peron}{University of Napoli ``Federico II'', Napoli, Italy}{}{}{}{}

 \authorrunning{ L.\ Bozzelli  and A.\ Peron}

 \Copyright{ Laura Bozzelli and Adriano Peron}


\nolinenumbers
\begin{document}

%
%

\maketitle              

\begin{abstract}
	Model checking (MC) for Halpern and Shoham's interval temporal logic \HS\ has been recently investigated in a systematic way, and it is known to be decidable under three distinct semantics (state-based, trace-based and tree-based semantics), all of them assuming \emph{homogeneity} in the propositional valuation.
Here, we focus on the \emph{trace-based semantics}, where the main semantic entities are the  infinite execution paths (traces) of the given Kripke structure. 
We introduce a quantitative extension of $\HS$ over traces, called \emph{Difference} $\HS$ ($\DHS$), allowing one to express timing constraints on the difference among interval lengths (\emph{durations}). 
We  show that MC and satisfiability of full $\DHS$ are in general undecidable, so, we investigate the decidability border for these 
problems by considering natural syntactical fragments of $\DHS$. In particular, we identify a maximal decidable fragment $\DHSS$ of $\DHS$ proving in addition that the considered problems for this fragment are at least  \TWOEXPSPACE-hard. Moreover, by exploiting new results on linear-time hybrid logics, we show that for an equally expressive fragment of
$\DHSS$, the problems are  \EXPSPACE-complete. Finally, we provide a characterization of $\HS$ over traces by means of the one-variable fragment of a novel hybrid logic.
\end{abstract}
\input{Introduction}
\input{Preliminaries}
\input{DifferenceHS}
\input{UndecidabilityDHS}

\input{DecisionProceduresDHS}

\input{Conclusion}

\bibliographystyle{plainurl}
\bibliography{bib2}
\newpage
\input{Appendix}
\end{document}

%% file: Macro.tex
\newcommand{\details}[1]{{}}
\newcommand{\tpl}[1]{(#1)}
\newcommand{\len}[1]{{\mathsf{len}_{#1}}}
\newcommand{\DefinedAs}{\ensuremath{\,\stackrel{\text{\textup{def}}}{=}\,}}

\newcommand{\Nat}{{\mathbb{N}}}
\newcommand{\INT}{{\mathbb{Z}}}
\newcommand{\NatP}{{\mathbb{N}_+}}

\newcommand{\Logic}{{\mathfrak{F}}}
\newcommand{\Prop}{\mathcal{AP}}
\newcommand{\Lang}{{\mathcal{L}}}
\newcommand{\LO}{\mathbb{U}}
\newcommand{\LOSup}{U}
\newcommand{\Intvs}{\mathbb{I}}
\newcommand{\IS}{\mathcal{S}}

\def\Bool{{\mathbb{B}}}

 \newcommand{\Au}{\ensuremath{\mathcal{A}}}
\newcommand{\Ku}{\ensuremath{\mathcal{K}}}
\newcommand{\Lab}{\mathit{Lab}}

\newcommand{\HS}{\text{\sffamily HS}}
\newcommand{\HL}{\text{\sffamily HL}}
\newcommand{\CHL}{\text{\sffamily CHL}}
\newcommand{\MCHL}{\text{\sffamily MCHL}}
\newcommand{\SHL}{\text{\sffamily SHL}}
\newcommand{\TwoAWA}{\text{\sffamily 2AWA}}

\newcommand{\SCHL}{\text{\sffamily SCHL}}
 \newcommand{\DHS}{\text{\sffamily DHS}}
\newcommand{\DHSS}{\text{\sffamily DHS}_{simple}}
\newcommand{\DHSFrag}{\text{\sffamily D}(\mathsf{X_1\ldots X_n})}
\newcommand{\DHSSFrag}{\text{\sffamily D}_{simple}(\mathsf{X_1\ldots X_n})}
\newcommand\DHSSF[1]{\text{\sffamily D}_{simple}(#1)}
\newcommand\DHSF[1]{\text{\sffamily D}(#1)}
\newcommand{\MTL}{\text{\sffamily MTL}}

\newcommand{\CTL}{\text{\sffamily CTL}}
\newcommand{\CTLStar}{\text{\sffamily CTL$^{*}$}}
\newcommand{\LTL}{\text{\sffamily LTL}}

\newcommand{\NWA}{\text{\sffamily NWA}}
\newcommand{\FO}{\text{\sffamily FO}}

\newcommand{\Frag}{\ensuremath{\mathcal{F}}}

\newcommand{\RelA}{\ensuremath{\mathcal{R}_A}}
\newcommand{\RelL}{\ensuremath{\mathcal{R}_L}}
\newcommand{\RelB}{\ensuremath{\mathcal{R}_B}}
\newcommand{\RelO}{\ensuremath{\mathcal{R}_O}}
\newcommand{\RelE}{\ensuremath{\mathcal{R}_E}}
\newcommand{\RelD}{\ensuremath{\mathcal{R}_D}}
\newcommand{\RelX}{\ensuremath{\mathcal{R}_X}}
\newcommand{\RelXt}{\ensuremath{\mathcal{R}_{\overline{X}}}}
\newcommand{\RelBt}{\ensuremath{\mathcal{R}_{\overline{B}}}}

 \newcommand{\RelXP}[1]{\ensuremath{\mathcal{R}_{X_{#1}}}}

\DeclareMathOperator{\hsX}{\langle X\rangle}
\DeclareMathOperator{\hsA}{\langle A\rangle}
\DeclareMathOperator{\hsL}{\langle L\rangle}
\DeclareMathOperator{\hsB}{\langle B\rangle}
\DeclareMathOperator{\hsE}{\langle E\rangle}
\DeclareMathOperator{\hsD}{\langle D\rangle}
\DeclareMathOperator{\hsO}{\langle O\rangle}

\DeclareMathOperator{\hsAt}{\langle \overline{A}\rangle}

\DeclareMathOperator{\hsBt}{\langle \overline{B}\rangle}

\DeclareMathOperator{\hsEt}{\langle \overline{E}\rangle}

\DeclareMathOperator{\hsDt}{\langle \overline{D}\rangle}

\DeclareMathOperator{\hsUX}{[X]}

\DeclareMathOperator{\hsUA}{[A]}

\DeclareMathOperator{\hsUB}{[B]}
\DeclareMathOperator{\hsUE}{[E]}
\DeclareMathOperator{\hsUD}{[D]}
\DeclareMathOperator{\hsUO}{[O]}

\newcommand{\AB}{\mathsf{AB}}

\newcommand{\ABE}{\mathsf{ABE}}

\newcommand{\BEBbarEbar}{\mathsf{BE\overline{B}\overline{E}}}
\newcommand{\BDEBbarDbarEbar}{\mathsf{BED\overline{B}\overline{D}\overline{E}}}

\newcommand{\B}{\mathsf{B}}

\newcommand{\BE}{\mathsf{BE}}

\newcommand{\ABBbar}{\mathsf{AB\overline{B}}}

\newcommand{\Always}{\textsf{G}}
\newcommand{\PAlways}{\textsf{H}}
\newcommand{\Eventually}{\textsf{F}}
\newcommand{\PEventually}{\textsf{P}}

\newcommand{\Swap}{\textsf{swap}}
\newcommand{\DBinder}{\text{$\downarrow$$x$}}
\newcommand{\Binder}[1]{\text{$\downarrow$$#1$}}

\def\NLOGSPACE{{\sc NLogspace}}

\def\PSPACE{{\sc Pspace}}
\def\EXPSPACE{{\sc Expspace}}
\def\TWOEXPSPACE{{\sc 2Expspace}}
\def\NEXPSPACE{{\sc NExpspace}}

\newcommand{\PTIME}{\mathbf{P}}
\newcommand{\NP}{\mathbf{NP}}
\newcommand{\coNP}{\mathbf{co-NP}}

\newcommand{\rec}{{\textit{rec}}}
\newcommand{\init}{{\textit{init}}}
\newcommand{\Inst}{{\textit{L}}}
\newcommand{\Inc}{{\textit{Inc}}}
\newcommand{\Inter}{{\textit{Int}}}
\newcommand{\Dec}{{\textit{Dec}}}
\newcommand{\Zero}{{\textit{Zero}}}
\newcommand{\LR}{{\mathsf{lr}}}
\newcommand{\inc}{{\mathsf{inc}}}
\newcommand{\dec}{{\mathsf{dec}}}
\newcommand{\zero}{{\mathsf{if\_zero}}}
\newcommand{\instr}{{\textit{op}}}

\newcommand{\con}{\mathsf{con}}
\newcommand{\dom}{\mathsf{dom}}

\newcommand{\MAX}{\textit{Max}}
\newcommand{\Left}{\textit{left}}
\newcommand{\Right}{\textit{right}}
\newcommand{\LNext}{\textit{left\_next}}
\newcommand{\RNext}{\textit{right\_next}}

\newcommand{\Instance}{\mathcal{I}}
\newcommand{\Init}{\textit{in}}
\newcommand{\acc}{\textit{acc}}
 \newcommand{\Down}{\textit{down}}
\newcommand{\Up}{\textit{up}}

\newcommand{\Succ}{\mathsf{Succ}}
\newcommand{\Atoms}{\mathsf{Atoms}}
\newcommand{\False}{\mathsf{false}}
\newcommand{\True}{\mathsf{true}}
\newcommand{\MNF}{\textit{MNF}}
\newcommand{\OP}{\mathsf{O}}

%% file: Introduction.tex
\section{Introduction}


\emph{Interval Temporal Logics} (ITLs, see~\cite{HS91,DigitalCircuitsThesis,Ven90}) provide an alternative setting for reasoning about time with respect to the more popular \emph{Point-based} Temporal Logics (PTLs) whose most known representatives are  the linear-time temporal logic $\LTL$~\cite{pnueli1977temporal} and the branching-time temporal logics $\CTL$ and $\CTLStar$~\cite{emerson1986sometimes}.
ITLs assume intervals, instead of points, as their primitive temporal entities
allowing one to specify temporal properties that involve, e.g., actions with duration, accomplishments, and temporal aggregations, which are inherently ``interval-based'', and thus cannot be naturally expressed by PTLs. 
The most prominent example of ITLs is \emph{Halpern and Shoham's modal logic of time intervals} (\HS)~\cite{HS91}
 featuring modalities for any Allen's
	relation~\cite{All83} (apart from equality).
The satisfiability problem for \HS\ turns out to be highly undecidable for all interesting (classes of) linear orders~\cite{HS91} both for the full logic and  most of its fragments~\cite{BresolinMGMS14,Lod00,MM14}.

Model checking (MC) of (finite) Kripke structures against \HS\ has been  investigated in recent papers~\cite{LM13,LM14,LM16,MolinariMMPP16,BozzelliMMPS18,BozzelliMMPS19,BozzelliMMPS19b,MolinariMP18,BozzelliMMP20,BozzelliMMPS22} which provide more encouraging results.
In the model checking setting, each finite path of a Kripke structure is an \emph{interval} having a labelling derived from the labelling of the component states: a proposition letter holds over an interval if and only if it holds over each component state (\emph{homogeneity assumption}~\cite{Roe80}).
Most of the results have been obtained by adopting the so-called \emph{state-based semantics}~\cite{MolinariMMPP16}:
intervals/paths are ``forgetful" of the history leading to their starting state, and time branches both in the future and in the past. In this setting, MC of full $\HS$ is decidable: the problem  is at least
\EXPSPACE-hard~\cite{BMMPS16}, while the only known upper bound is non-elementary~\cite{MolinariMMPP16}.
The known complexity bounds  for full $\HS$ coincide with those for the linear-time fragment $\BE$ of $\HS$ which features modalities $\hsB$  and $\hsE$ for prefixes and suffixes.
 Whether or not model checking for $\BE$ can be solved elementarily is a difficult open question.
On the other hand, in the state-based setting, the exact complexity of MC for many meaningful (linear-time or branching-time) syntactic
 fragments of $\HS$, which ranges from $\coNP$ to $\PTIME^{\NP}$, \PSPACE, and beyond, has been determined in a series of papers~\cite{BozzelliMMPS18,BozzelliMMPS19b,BozzelliMP21,BozzelliMPS21,BozzelliMMPS22}.
\newline
The expressiveness of $\HS$ with the state-based semantics  has been studied in~\cite{BozzelliMMPS19}, together  with other two decidable variants:
	the \emph{computation-tree-based semantics}   and the \emph{traces-based} one. For the first variant, past is linear: each
interval may have several possible future, but only a unique past. Moreover, past is finite and cumulative, and is never forgotten.
The trace-based approach instead relies on a linear-time setting, where the infinite paths (traces) of the given Kripke structure are the main semantic entities.
It is known that the computation-tree-based variant of $\HS$ is expressively equivalent to finitary  $\CTLStar$ (the variant of $\CTLStar$ with quantification over finite paths), while the trace-based variant is equivalent to $\LTL$~\cite{BozzelliMMPS19}. The state-based variant is more expressive than the computation-tree-based variant and expressively incomparable with both $\LTL$ and $\CTLStar$~\cite{BozzelliMMPS19}.  
\newline
\hspace{0.2cm} In this paper, we introduce a quantitative extension of the interval temporal logic $\HS$ under the \emph{trace-based semantics}, called
\emph{Difference} $\HS$ ($\DHS$). The extension is obtained by means of equality and inequality constraints on the temporal modalities  which allow  to specify integer bounds on the difference between the durations (lengths) of the current interval and the interval selected by the modality. The logic $\DHS$ can also encode in a succinct way constraints on the duration of the  current interval. 
Thus, the considered framework non-trivially generalizes well-known discrete-time quantitative extensions of standard $\LTL$, such as  Metric Temporal Logic (\MTL)~\cite{Koymans90}, where one can essentially  express integer bounds on the duration of the interval having as endpoints the current position and the one selected by the temporal modality. 
\newline
We prove that MC and satisfiability of full $\DHS$ are in general undecidable. 
Thus, we investigate the decidability border of these problems by considering the syntactical fragments of $\DHS$ obtained by restricting the set of allowed constrained modalities.  In particular, we prove that for the syntactical fragment, namely $\DHSS$, whose constrained temporal modalities are associated with the  Allen's relations  subsuming the subset relation or its inverse,
the  problems are  decidable, though at least \TWOEXPSPACE-hard. On the other hand,
 we show that any constrained modality not supported by 
 $\DHSS$ is inherently problematic, since the addition to $\HS$ of such a modality  leads to undecidability. These results are a little surprising since it is well-known that under the adopted strict semantics admitting singleton intervals, all temporal modalities in $\HS$ can be expressed in terms of the ones associated with the Allen's relations subsuming the subset relation or its inverse.
Additionally, we identify an expressively complete fragment of $\DHSS$ for which satisfiability and model checking are shown to be \EXPSPACE-complete. The upper bound in \EXPSPACE\ is obtained by an elegant automaton-theoretic approach which exploits as a preliminary
step a linear-time translation of the  fragment of $\DHSS$ into a quantitative extension of the one-variable fragment of \emph{linear-time hybrid logic} $\HL$~\cite{FRS03,SW07,BozzelliL10}. Finally,  we provide a characterization of $\HS$ over traces in terms of a novel hybrid logic, namely $\SHL_1$, which lies between the one-variable and the two-variable fragment of
$\HL$. We prove that there are linear-time translations from $\HS$ formulas  into equivalent formulas of $\SHL_1$, and vice versa.

%% file: Preliminaries.tex
\section{Preliminaries}\label{sec:preliminary}

We fix the following notation. Let $\INT$ be the set of integers, $\Nat$ the set of natural numbers, and $\NatP\DefinedAs \Nat\setminus \{0\}$.
For a  finite or infinite word $w$ over some alphabet, $|w|$ denotes the length of $w$
($|w|=\infty$ if $w$ is infinite) and for all $0\leq i<|w|$, $w(i)$ is the
$(i+1)$-th letter of $w$.

  We fix a finite set $\Prop$ of atomic propositions. A \emph{trace} is an infinite word over $2^{\Prop}$.

Let $\Logic$ and $\Logic'$ be two logics interpreted over traces.
For a formula $\varphi\in\Logic$,  $\Lang(\varphi)$ denotes the
set of traces satisfying $\varphi$. 
Given  $\varphi\in\Logic$ and  $\varphi'\in\Logic'$,   $\varphi$ and $\varphi'$ are \emph{equivalent} if $\Lang(\varphi)=\Lang(\varphi')$.  The satisfiability problem for $\Logic$ is checking for a given  formula
$\varphi\in\Logic$, whether $\Lang(\varphi)\neq \emptyset$.\vspace{0.1cm}

\noindent\textbf{Kripke Structures.}
\details{In the context of model checking, finite state systems are usually modelled as finite Kripke  structures over a finite set \Prop$ of atomic
propositions which represent predicates over the states of the system.}
A \emph{(finite) Kripke structure} over   $\Prop$   is a tuple  $\Ku=\tpl{\Prop,S, E,\Lab,s_0}$, where  $S$ is a finite set of states,
$E\subseteq S\times S$ is a left-total transition relation, $\Lab:S\mapsto 2^{\Prop}$ is a labelling function assigning to each state $s$ the set of propositions that hold over it, and $s_0\in S$ is the initial state.
 An infinite  path $\pi$ of $\Ku$ is an infinite word over $S$ such that $\pi(0)=s_0$  and $(\pi(i),\pi(i+1))\in E$ for all $i\geq 0$. A finite path of $\Ku$ is a non-empty infix of some infinite path of $\Ku$.
An infinite  path $\pi$ induces the  trace given by
$\Lab(\pi(0))\Lab(\pi(1))\ldots$. We denote  by $\Lang(\Ku)$ the set of traces associated with the infinite paths of $\Ku$.
For a logic $\Logic$ interpreted  over traces, the \emph{model checking} (MC) \emph{problem against $\Logic$} is checking for a given  Kripke structure $\Ku$ and a formula $\varphi\in \Logic$, whether $\Lang(\Ku)\subseteq \Lang(\varphi)$.

\subsection{Allen's relations and Interval Temporal Logic $\HS$}

An interval algebra to reason about intervals and their relative orders was proposed by Allen in~\cite{All83}, while a systematic logical study of interval representation and reasoning was done a few years later by Halpern and Shoham, who introduced the interval temporal logic $\HS$ featuring one modality for each Allen
relation, but equality~\cite{HS91}.

Let $\LO = \tpl{\LOSup,<}$ be a  linear order over the nonempty set $\LOSup\neq \emptyset$, and $\leq$ be the reflexive closure of $<$. Given two elements $x,y\in \LOSup$ such that $x\leq y$, we denote by $[x,y]$
the (non-empty closed) \emph{interval} over $\LOSup$ given by the set of elements $z\in \LOSup$ such that $x\leq z$ and $z\leq y$.
We denote the set of all intervals over
$\LO$ by $\Intvs(\LO)$. Table~\ref{allen} gives a graphical representation of the Allen's relations $\RelA$, $\RelL$, $\RelB$, $\RelE$,
$\RelD$, and $\RelO$ for the given linear order together with the corresponding $\HS$ (existential) modalities.
For each $X\in \{A,L,B,E,D,O\}$, the Allen's relation $\RelXt$ is defined as the inverse of relation $\RelX$, i.e.
      $[x, y]\,\RelXt\,  [v, z]$ if $[v, z]\RelX   [x, y]$.
\details{
 We now recall the 13 Allen's relations but the equality over intervals of the given linear order $\LO = \tpl{\LOSup,<}$:
\begin{enumerate}
  \item the \emph{meet} relation $\RelA$, defined
    by $[x, y]\,\RelA\,  [v, z]$ if $y=v$ (i.e., the start-point of the second interval coincides
    with the end-point of the first interval);
      \item the \emph{before} relation $\RelL$, defined
    by $[x, y]\,\RelL \, [v, z]$ if $y<v$ (i.e., the start-point of the second interval strictly follows
   the end-point of the first interval);
    \item the \emph{started-by} relation $\RelB$, defined
    by $[x, y]\,\RelB\,  [v, z]$ if $x =v$ and $z<y$ (i.e., the  second interval is a proper prefix
   of the first interval);
   \item the \emph{finished-by} relation $\RelE$, defined
    by $[x, y]\,\RelE\,  [v, z]$ if $y =z$ and $x<v$ (i.e., the  second interval is a proper suffix
   of the first interval);
    \item the \emph{contains} relation $\RelD$, defined
    by $[x, y]\,\RelD\,  [v, z]$ if $x< v$ and $z<y$ (i.e., the  second interval is contained in the internal of the first interval);
       \item the \emph{overlaps} relation $\RelO$, defined
    by $[x, y]\,\RelO\,  [v, z]$ if $x<v<y<z$ (i.e., the second interval overlaps at the right the first interval);
    \item for each $X\in \{A,L,B,E,D,O\}$ the relation $\RelXt$, defined as the inverse of relation $\RelX$, i.e.
      $[x, y]\,\RelXt\,  [v, z]$ if $[v, z]\RelX   [x, y]$.
\end{enumerate}

Table~\ref{allen} gives a graphical representation of the Allen's relations $\RelA$, $\RelL$, $\RelB$, $\RelE$,
$\RelD$, and $\RelO$ together with the corresponding $\HS$ (existential) modalities.
}
\begin{table}[tb]
\centering
\caption{Allen's relations and corresponding $\HS$ modalities.}\label{allen}
\vspace*{-0.1cm}
\resizebox{\width}{0.8\height}
{
\begin{tabular}{cclc}
\hline
\rule[-1ex]{0pt}{3.5ex} Allen relation & $\HS$ & Definition w.r.t. interval structures &  Example\\
\hline

&   &   & \multirow{7}{*}{\input{allensRels.tex}}\\

\textsc{meets} & $\hsA$ & $[x,y]\,\RelA\,[v,z]\iff y=v$ &\\

\textsc{before} & $\hsL$ & $[x,y]\,\RelL\,[v,z]\iff y<v$ &\\

\textsc{started-by} & $\hsB$ & $[x,y]\,\RelB\,[v,z]\iff x=v\wedge z<y$ &\\

\textsc{finished-by} & $\hsE$ & $[x,y]\,\RelE\,[v,z]\iff y=z\wedge x<v$ &\\

\textsc{contains} & $\hsD$ & $[x,y]\,\RelD\,[v,z]\iff x<v\wedge z<y$ &\\

\textsc{overlaps} & $\hsO$ & $[x,y]\,\RelO\,[v,z]\iff x<v<y<z$ &\\

\hline
\end{tabular}}\vspace{-0.1cm}
\end{table}

$\HS$ formulas  $\varphi$ over $\Prop$
are defined as follows:\vspace{0.1cm}

$
   \varphi ::= \top \;\vert\;   p \;\vert\; \neg\varphi \;\vert\; \varphi \wedge \varphi \;\vert\; \hsX \varphi
$\vspace{0.1cm}

\noindent  where $p\in\Prop$ and $\hsX$ is the existential temporal modality  for the   (non-trivial)
Allen's relation $\RelX$, where $X\in\{A,L,B,E,D,O,\overline{A},\overline{L},\overline{B},\overline{E},\overline{D},\overline{O}\}$.
The size $|\varphi|$ of a formula $\varphi$ is the number of distinct subformulas of $\varphi$.
We  also exploit the standard logical connectives 
$\vee$   and $\rightarrow$   as abbreviations,
and for any temporal  modality $\hsX$, the dual universal modality $\hsUX$  defined as: $\hsUX\psi\DefinedAs \neg\hsX\neg\psi$.
Given any subset of Allen's relations $\{\RelXP{1},..,\RelXP{n}\}$, we denote by $\mathsf{X_1 \cdots X_n}$ the \HS\ fragment featuring temporal modalities for $\RelXP{1},..,\RelXP{n}$ only.

The   logic $\HS$ is interpreted on \emph{interval structures} $\IS=\tpl{\Prop,\LO,\Lab}$, which are linear orders $\LO$ equipped with a labelling function $\Lab: \Intvs(\LO) \to 2^{\Prop}$ assigning to each interval the set of propositions that hold over it. Given an $\HS$ formula $\varphi$ and an interval $I  \in \Intvs(\LO)$, the satisfaction relation $I\models_\IS \varphi$, meaning that $\varphi$ holds at the interval $I$ of $\IS$, is inductively defined as follows (we omit the semantics of the Boolean connectives which is standard):\vspace{0.1cm}

 $
 \begin{array}{ll}
I \models_\IS  p  &  \Leftrightarrow  p\in \Lab(I)  \\
I\models_\IS     \hsX \varphi  &  \Leftrightarrow \text{there is an interval $J\in \Intvs(\LO)$ such that $I\, \RelX \,J$ and }   J\models_\IS \varphi
\end{array}
 $\vspace{0.1cm}

\noindent It is worth noting that  we assume the \emph{non-strict semantics of $\HS$},
which admits intervals  consisting of a single point. Under such an assumption, all $\HS$-temporal modalities  can be expressed in terms of $\hsB, \hsE, \hsBt$, and $\hsEt$~\cite{Ven90}.
As an example,  $\hsD\varphi$ can be expressed in terms of $\hsB$ and $\hsE$ as $ \hsB\hsE\varphi$, while  $\hsA\varphi$ can be expressed in terms of $\hsE$ and $\hsBt$ as
$ (\hsUE\,\neg\top \wedge (\varphi \vee \hsBt \varphi)) \vee \hsE (\hsUE\,\neg\top \wedge (\varphi \vee \hsBt \varphi))
 $.\vspace{0.1cm}

\noindent \textbf{Interpretation of $\HS$ over traces.}
In this paper, we focus on interval structures $\IS=\tpl{\Prop,(\Nat,<),\Lab}$ over the standard linear order on $\Nat$ ($\Nat$-interval structures for short) satisfying the \emph{homogeneity principle}: a proposition holds over an interval if
and only if it holds over all its subintervals. Formally, $\IS$ is \emph{homogeneous} if
 for every interval $[i,j]$ over $\Nat$ and every $p \in \Prop$,
it holds that $p\in \Lab([i,j])$ if and only if $p\in \Lab([h,h])$
for every $h\in [i,j]$.
Note that homogeneous $\Nat$-interval structures over $\Prop$ correspond to traces where, intuitively, each interval is mapped to an infix of the trace. Formally, each trace $w$  induces the homogeneous $\Nat$-interval structure $\IS(w)$ whose labeling function $\Lab_w$ is defined as follows:
for all $i,j\in\Nat$ with $i\leq j$ and $p\in\Prop$,  $p\in \Lab_w([i,j])$ if and only if $p\in w(h)$ for all $h\in [i,j]$. 
This mapping from traces to homogeneous $\Nat$-interval structures over $\Prop$ is evidently a bijection.
For a trace $w$, an interval $I$ over $\Nat$, and an $\HS$ formula $\varphi$, we write $I\models_w \varphi$ to mean that $I\models_{\IS(w)} \varphi$. The trace $w$ satisfies $\varphi$, written $w\models \varphi$, if $[0,0]\models_w \varphi$.
For an interval $I=[i,j]$ over $\Nat$, we denote by $|I|$ the length of $I$, given by $j-i+1$.

It is known that $\HS$ over traces has the same expressiveness as standard $\LTL$~\cite{BozzelliMMPS19}, where the latter is expressively complete
for standard first-order logic $\FO$ over traces~\cite{kamp1968tense}. In particular, the fragment $\AB$ of $\HS$ is sufficient for capturing full $\LTL$~\cite{BozzelliMMPS19}: given an $\LTL$ formula, one can construct in linear-time an equivalent
$\AB$ formula~\cite{BozzelliMMPS19}.  Note that when interpreted on infinite words $w$, modality
 $\hsB$   allows to select proper non-empty prefixes   of the current infix subword of $w$, while modality
  $\hsA$   allows to select subwords whose first position coincides with the last position of the current interval. For each $k\geq 1$, we denote by $\len{k}$  the $\B$ formula capturing the intervals of length $k$: $\len{k}\DefinedAs (\underbrace{\hsB \ldots \hsB}_{\text{$k-1$ times}}\top)\wedge (\underbrace{\hsUB \ldots \hsUB}_{\text{$k$ times}}\neg\top)$. \vspace{-0.2cm} 

%% file: allensRels.tex
\begin{tikzpicture}[scale=0.785]
\draw[draw=none,use as bounding box](-0.3,0.2) rectangle (3.3,-3.1);
\coordinate [label=left:\textcolor{red}{$x$}] (A0) at (0,0);
\coordinate [label=right:\textcolor{red}{$y$}] (B0) at (1.5,0);
\draw[red] (A0) -- (B0);
\fill [red] (A0) circle (2pt);
\fill [red] (B0) circle (2pt);

\coordinate [label=left:$v$] (A) at (1.5,-0.5);
\coordinate [label=right:$z$] (B) at (2.5,-0.5);
\draw[black] (A) -- (B);
\fill [black] (A) circle (2pt);
\fill [black] (B) circle (2pt);

\coordinate [label=left:$v$] (A) at (2,-1);
\coordinate [label=right:$z$] (B) at (3,-1);
\draw[black] (A) -- (B);
\fill [black] (A) circle (2pt);
\fill [black] (B) circle (2pt);

\coordinate [label=left:$v$] (A) at (0,-1.5);
\coordinate [label=right:$z$] (B) at (1,-1.5);
\draw[black] (A) -- (B);
\fill [black] (A) circle (2pt);
\fill [black] (B) circle (2pt);

\coordinate [label=left:$v$] (A) at (0.5,-2);
\coordinate [label=right:$z$] (B) at (1.5,-2);
\draw[black] (A) -- (B);
\fill [black] (A) circle (2pt);
\fill [black] (B) circle (2pt);

\coordinate [label=left:$v$] (A) at (0.5,-2.5);
\coordinate [label=right:$z$] (B) at (1,-2.5);
\draw[black] (A) -- (B);
\fill [black] (A) circle (2pt);
\fill [black] (B) circle (2pt);

\coordinate [label=left:$v$] (A) at (1.3,-3);
\coordinate [label=right:$z$] (B) at (2.3,-3);
\draw[black] (A) -- (B);
\fill [black] (A) circle (2pt);
\fill [black] (B) circle (2pt);

\coordinate (A1) at (0,-3);
\coordinate (B1) at (1.5,-3);
\draw[dotted] (A0) -- (A1);
\draw[dotted] (B0) -- (B1);
\end{tikzpicture}

%% file: DifferenceHS.tex
\section{Difference Interval Temporal Logic}\label{sec:definitionDHS}

 In this section, we introduce a quantitative extension of the  logic $\HS$ under the trace-based semantics, we call \emph{Difference $\HS$} ($\DHS$ for short). The extension is obtained by means of equality and inequality constraints on the temporal modalities of $\HS$ which allow  to compare the difference between the length of the interval selected by the temporal modality and the length of the current interval with an integer constant.

The set of $\DHS$ formulas $\varphi$ over  $\Prop$  is inductively defined as follows:\vspace{0.1cm}

$
\varphi ::= \top  \;\vert\; p \;\vert\; \neg\varphi \;\vert\; \varphi \wedge \varphi \;\vert\; \hsX \varphi \;\vert\; \hsX_{\Delta \sim c} \varphi
$\vspace{0.1cm}

\noindent where $p\in\Prop$, $\sim \in \{<,\leq,=,>,\geq\}$, $c\in \INT$, and $\hsX_{\Delta\sim c}$ is the  existential \emph{constrained} temporal modality  for the
Allen's relation  $\RelX$ where $X\in\{A,L,B,E,D,O,\overline{A},\overline{L},\overline{B},\overline{E},\overline{D},\overline{O}\}$.
We exploit the symbol $\Delta$ in $\hsX_{\Delta\sim c}$ to emphasize that the constraint $\sim c$ refer to the difference between the lengths of two intervals, the one selected by the modality $\hsX$ and the current one. A formula $\varphi$ is \emph{monotonic} if it does not use equality constraints $\Delta= c$ as subscripts of the temporal modalities.
For any constrained  modality $\hsX_{\Delta\sim c}$, the dual universal modality $\hsUX_{\Delta\sim c}$  is an abbreviation for $\neg\hsX_{\Delta\sim c}\neg\varphi$. We assume that the constants $c$ in the difference constraints are encoded in binary. Thus, the size $|\varphi|$ of a $\DHS$ formula $\varphi$ is defined as the number of distinct subformulas of $\varphi$ multiplied the number of bits for encoding the maximal constant occurring in $\varphi$.  The semantics of the constrained modalities is as follows:
\begin{compactitem}
\item $I\models_w \hsX_{\Delta\sim c}\varphi$ \,$ \Leftrightarrow$\, for some interval $J$  such that $I\, \RelX \,J$ and  $|J|-|I|\sim c$: $ J\models_w \varphi$.
\end{compactitem}\vspace{0.2cm}

\noindent We consider the following fragments of $\DHS$:
\begin{compactitem}
  \item The fragment $\DHSS$ which disallows constrained modalities for the Allen's relations $\RelA$, $\RelL$, $\RelO$, and their inverses, and for any Allen's relation $\RelX$, the fragment $\DHS_X$ allowing constrained modalities for the Allen's relation
  $\RelX$ only.
  \item For any subset of Allen's relations  $\{\RelXP{1},..,\RelXP{n}\}$, the fragment $\DHSFrag$
featuring   temporal modalities for $\RelXP{1},..,\RelXP{n}$  only, and the common fragment of
$\DHSFrag$ and $\DHSS$, denoted by $\DHSSFrag$.
\end{compactitem}\vspace{0.1cm}

\noindent \textbf{Expressiveness issues.} As mentioned in Section~\ref{sec:preliminary}, all the temporal modalities of $\HS$ can be expressed in terms of  $\hsB, \hsE, \hsBt$, and $\hsEt$.
In the considered quantitative setting, these  interdefinability results cannot be generalized to the constrained versions of the temporal modalities. In particular, we will show in
Section~\ref{DecisionProcedures} that the fragment $\DHSS$ of $\DHS$, featuring constrained modalities only for the Allen's relations
$\RelB$, $\RelD$, $\RelE$, and their inverses, is not more expressive than $\HS$. On the other hand, we will establish in Section~\ref{sec:undecidabilityDHS}  that the fragments $\DHS_X$, where  $X\in\{A,L,O,\overline{A},\overline{L},\overline{O}\}$,
are highly undecidable.

Unlike $\HS$, in $\DHSS$  we can succinctly express that an arbitrary $\HS$ property $\varphi$ holds 
in the \emph{maximal proper sub-intervals} of the current non-singleton interval by the formula 
$(\hsE_{\Delta \geq -1}\varphi)\wedge (\hsB_{\Delta \geq -1}\varphi)$.
Moreover, we can succinctly encode constraints on the length of the current interval.
For  an integer $n>0$, the  $\DHSS$ formula $\hsB_{\Delta \leq -n+1}\top$ (resp., $\neg\hsB_{\Delta \leq -n}\top$)
 characterizes the intervals of length at least (resp., at most) $n$.

 \input{FigureScheduler.tex}

\begin{example}\label{ex:krypke}
We consider the behaviour of a scheduler serving $N$ processes which continuously request 
the use of a common resource. The behaviour of each process $P_i$, with $1\leq i \leq N$, is represented by the Kripke structure $\Ku_{P_i}$, depicted in Figure~\ref{KSched},
whose atomic propositions $p^i_I$, $p^i_R$, and $p^i_U$ label the states where the process is idling, requests the resource, and uses the resource, respectively. The behaviour of the scheduler $H$ is modeled by the Kripke structure $\Ku_{H}$ in Figure~\ref{KSched} whose propositions $q_I$, $q_{U_1},\ldots q_{U_N}$ 
label the states  where $H$ is idling or assigns the resource to the $i$-th process (proposition $q_{U_i}$). The considered Kripke structure $\Ku_{Sched}$, depicted in Figure~\ref{KSched} for $N=2$, is the cartesian product of the Kripke structures $\Ku_{P_1},\ldots,\Ku_{P_N}, \Ku_{H}$ 
with the additional requirement that the scheduler is in state $w_i$ \emph{iff} the $i$-th process is in state $v_2$. The set of atomic propositions labelling each compound state is the union of the sets of propositions labelling the component states. 

As an example of specification, we consider the requirement that the $i$-th process can unsuccessfully iterate a request (i.e., without finally having the resource granted) for an interval of at least $m$ and at most $M$ time units. This can be expressed in
 $\DHSS$ as:\vspace{0.1cm} 

$\hsUA  \hsUA  [(\MAX_{p^i_{R}} \wedge \hsBt_{\Delta \leq 1} \hsE p^i_I)  \,\rightarrow\, (\hsB_{\Delta \leq -m+1}\top\wedge \neg\hsB_{\Delta \leq -M}\top)]$\vspace{0.1cm}

\noindent where for a proposition $p$, $\MAX_{p}\DefinedAs p \wedge (\neg \hsBt p) \wedge (\neg \hsEt p) $ captures the maximal length intervals where $p$ homogeneously holds. Note  that $\hsBt_{\Delta \leq 1} \hsE p^i_I$ ensures that the maximal homogeneous interval where $p^i_{R}$ holds is followed by a $p^i_I$-state.

\details{
In order to check properties of $\Ku_{Sched}$, 
it is crucial to capture maximal length intervals where a designated atomic proposition $p$ homogeneously holds (for instance, the interval during which a process requires or uses the resource).  \emph{Maximal $p$-homogeneous intervals} can be easily characterized
in  $\DHSS$ by the following formula $\MAX_p$: \vspace{0.1cm}
$
\MAX_p \DefinedAs p \wedge (\hsBt_{\Delta \leq 1} \neg p) \wedge (\hsEt_{\Delta \leq 1} \neg p)
$\vspace{0.1cm}

\noindent For  an integer $n>0$, the monotonic $\DHSS$ formula $\hsB_{\Delta \leq -n+1}\top$ (resp., $\neg\hsB_{\Delta \leq -n}\top$)
 succinctly  captures the intervals of length at least (resp., at most) $n$.
%
The previous  formulas can be used to require
that the $i$-th process can unsuccessfully iterate a request (i.e., without finally having the resource granted) for an interval of at least $m$ and at most $M$ time units: $\hsUD [(\MAX_{p^i_{R}} \wedge \hsBt_{\Delta = 1} \hsE p^i_I)  \,\rightarrow\, (\hsB_{\Delta \leq -m+1}\top\wedge \neg\hsB_{\Delta \leq -M}\top)]$. 
Note  that $\hsBt_{\Delta = 1} \hsE p^i_I$ ensures that the maximal homogeneous interval where $p^i_{R}$ holds is followed by a
$p^i_I$-state.} 
\end{example}

%% file: FigureScheduler.tex
 \begin{figure}
 
\centering
\vspace*{-1.2cm}
\resizebox{0.8\width}{0.8\height}{
\begin{tikzpicture}[->,>=stealth',shorten >=1pt,auto,node distance=2.2cm,thick,main node/.style={circle,draw}]

  \node[main node,style={double}] (1) {$\stackrel{v_0}{p^i_I}$};
  \node[main node,fill=gray!35] (3) [below=0.7cm of 1] {$\stackrel{v_1}{p^i_R}$};
  \node[main node,fill=gray!35] (6) [below of=3] {$\stackrel{v_2}{p^i_U}$};
\node (8b) [above=0.2cm of 1,draw = none] {};
\node (8) [left=0.3cm of 8b] {$\Ku_{P_i}$};

  \path[every node/.style={font=\small}]
    (1) edge node {} (3)
    (3) edge node {} (6)
        edge [bend left] node[right] {} (1)
    (6) edge [bend left] node[right] {} (1)
     %
     (1) edge [out=90,in=0,looseness=5] node[very near end,left] {} (1)
     (3) edge [out=90,in=0,looseness=5] node[very near end,left] {} (3)
     (6) edge [out=90,in=0,looseness=5] node[very near end,left] {} (6)
    ;
%
\node[main node,fill=gray!50] (10) [right=2cm of 3] {$\stackrel{w_1}{q_{U_1}}$};
\node (12) [right of=10] {\ldots};
 \node[main node,style={double}] (11)[above=1.1cm of 12] {$\stackrel{w_0}{q_I}$};
\node[main node,fill=gray!50] (13) [right of=12] {$\stackrel{w_N}{q_{U_N}}$};
\node (14b) [above=0.2cm of 11,draw = none] {};
\node (14) [left=0.3cm of 14b] {$\Ku_H$};
\path[every node/.style={font=\small}]
    (11) edge node {} (10)
         edge node {} (13)
     (11) edge [out=90,in=0,looseness=5] node[very near end,left] {} (11)
     (13) edge [out=-90,in=0,looseness=5] node[very near end,right] {} (13)
     (10) edge [out=-90,in=0,looseness=5] node[very near end,right] {} (10)
     (10) edge [bend left] node[right] {} (11)
     (13) edge [bend left] node[right] {} (11)
;
%
\node[main node,fill=gray!50] (20) [right=2cm of 13] {$v_1v_0w_0$};
\node[main node,fill=gray!50] (21) [right of=20] {$v_0v_1w_0$};
\node (25) [right of=21] {\ldots};
 \node[main node,style={double}] (22)[above=0.7cm of 21] {$v_0v_0w_0$};
 \node[main node,fill=gray!50] (23) [below of=20] {$v_2v_0w_1$};
 \node[main node,fill=gray!50] (24) [below of=21] {$v_0v_2w_2$};
 \node (26) [right of=24] {\ldots};
 \node (27b) [above=0.2cm of 22,draw = none] {};
 \node (27) [left=0.3cm of 27b] {$\Ku_{Sched}$};
 \path[every node/.style={font=\small}]
    (22) edge node {} (20)
         edge node {} (21)
    (20) edge node {} (23)
    (21) edge node {} (24)
         edge node {} (26)
    (22) edge [out=90,in=0,looseness=5] node[very near end,left] {} (22)
     (20) edge [out=90,in=0,looseness=5] node[very near end,left] {} (20)
     (21) edge [out=90,in=0,looseness=5] node[very near end,right] {} (21)
      (23) edge [out=90,in=0,looseness=5] node[very near end,right] {} (23)
(24) edge [out=90,in=0,looseness=5] node[very near end,right] {} (24)
(24) edge [bend left] node[right] {} (22)
(23) edge [bend left = 70] node[right] {} (22)
(21) edge [bend left] node[right] {} (22)
(20) edge [bend left] node[right] {} (22)
;
\end{tikzpicture}}
\caption{The Kripke structure $\Ku_{Sched}$ for two processes.}\label{KSched}
\vspace{-0.4cm}
\end{figure}
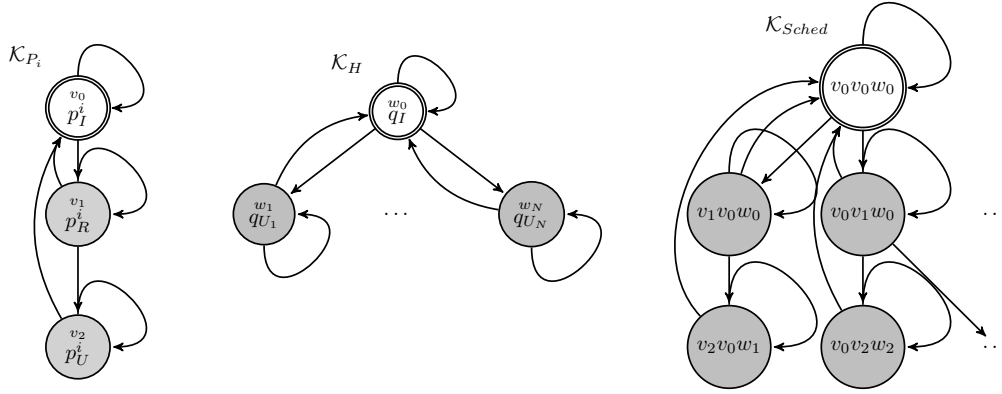 

%% file: UndecidabilityDHS.tex
\section{Undecidability of $\DHS$}\label{sec:undecidabilityDHS}

In this section, we establish that model checking and satisfiability for the novel logic $\DHS$ are highly  undecidable even
for the fragments $\DHS_X$, where  $X\in \{A,L,O,\overline{A},\overline{L},\overline{O}\}$.

\begin{theorem}\label{theorem:undecidabilityDHS}
Model checking and satisfiability for the fragment  $\DHS_X$ of $\DHS$, where $X\in \{A,L,O,\overline{A},\overline{L},\overline{O}\}$,  are
$\Sigma_{1}^{1}$-hard even if the unique constant used in the constraints is $0$, and in case
$X\in \{A, O,\overline{A}, \overline{O}\}$ even if the unique exploited constraint is $\geq 0$ (or, dually, $\leq 0$).
\end{theorem}

We prove Theorem~\ref{theorem:undecidabilityDHS} for the part concerning the satisfiability problem for the fragments $\DHS_A$, $\DHS_L$, and $\DHS_O$
 (the parts  for the model checking problem and for the fragments $\DHS_{\overline{A}}$, $\DHS_{\overline{L}}$, and $\DHS_{\overline{O}}$ being similar). We provide polynomial-time reductions 
 from the \emph{recurrence problem} of
\emph{non-deterministic  Minsky $2$-counter machines}~\cite{Harel86}.
Fix such a machine which is a tuple $M = \tpl{Q,\Delta,\delta_\init,\delta_\rec}$,
where  $Q$ is a finite set of (control) locations,
 $\Delta \subseteq Q\times \Inst \times Q$ is a transition relation over the instruction set $\Inst= \{\inc,\dec,\zero\}\times \{1,2\}$, and $\delta_\init\in \Delta$ and $\delta_\rec\in \Delta$ are two designated transitions, the initial and the recurrent one.
For each counter $c\in\{1,2\}$, let $\Inc(c)$ (resp., $\Dec(c)$, resp., $\Zero(c)$) be the
set of transitions $\delta\in\Delta$ whose instruction is $(\inc,c)$ (resp., $(\dec,c)$, resp., $(\zero,c)$).

An $M$-configuration is a pair $(\delta,\nu)$ consisting of a transition $\delta\in \Delta$ and a counter valuation $\nu: \{1,2\}\to \Nat$. A  computation of $M$ is an \emph{infinite} sequence of configurations of the form $((q_0,(\instr_0,c_0),q_1),\nu_0),((q_1,(\instr_1,c_1),q_2),\nu_1),\ldots$  such that for each $i\geq 0$:
\begin{inparaenum}[(i)]
  \item    $\nu_{i+1}(3-c_i)= \nu_i(3-c_i)$;
  \item  $\nu_{i+1}(c_i)= \nu_i(c_i) +1$ if $\instr_i=\inc$;
  \item $\nu_{i+1}(c_i)= \nu_i(c_i) -1$ if $\instr_i=\dec$; and 
   \item  $\nu_{i+1}(c_i)= \nu_i(c_i)=0$ if $\instr_i= \zero$.
\end{inparaenum}
A \emph{recurrent computation} is a computation starting at the initial configuration $(\delta_\init,\nu_0)$, where $\nu_0(c)=0$ for each $c\in \{1,2\}$, which visits the 
transition $\delta_\rec$ infinitely often.
 The \emph{recurrence problem} is to decide whether for the given machine $M$, there is a recurrent computation.  This problem is known to be $\Sigma_{1}^{1}$-complete~\cite{Harel86}.


For each $X\in \{A,L,O\}$, we construct a  $\DHS_X$ formula
$\varphi_{M,X}$   such that $M$ has a recurrent computation \emph{iff} $\varphi_M$ is satisfiable.
The reduction for the fragment  $\DHS_L$, given in the following,   is quite different from the ones for the fragments
$\DHS_A$ and  $\DHS_O$  which are given in Appendix~\ref{APP:UndecisabilityDHSForAO}. Indeed, while the quantitative versions of   modalities $\hsA$ and $\hsO$ allow to impose quantitative constraints on adjacent encodings of $M$-configurations, this is not possible for the quantitative version  of modality $\hsL$  whose semantics is not ``local", and for this modality, a different encoding of the computations of $M$ is required.

We exploit some auxiliary $\DHS$ formulas, where $\psi$ is an arbitrary $\DHS$ formula.
Formulas $\Left(\psi)$ and $\Right(\psi)$ assert that $\psi$ holds at the singular intervals corresponding to the left and right endpoints, respectively, of the current  interval.\vspace{0.1cm}

$
 \Left(\psi) \DefinedAs  (\len{1} \wedge \psi)\vee \hsB (\len{1}\wedge \psi) \quad\quad
 \Right(\psi) \DefinedAs \hsA(\len{1} \wedge \psi)
$\vspace{0.1cm}

\noindent For the current interval $[i,j]$, $\RNext(\psi)$ (resp., $\LNext(\psi)$) asserts that
$\psi$ holds at the singleton interval $[j+1,j+1]$ (resp., $[i+1,i+1]$), while $\Inter(\psi)$  requires that  there is an internal position $i<h<j$ such that $\psi$ holds at the singleton interval $[h,h]$.\vspace{0.1cm}
 
$
 \RNext(\psi) \DefinedAs  \hsA(\len{2} \wedge \hsA(\len{1}\wedge \psi)) \quad\quad
 \LNext(\psi) \DefinedAs \Left(\RNext( \psi))
$\vspace{0.1cm}
 
$
\Inter(\psi) \DefinedAs \hsB(\neg\len{1}\wedge \Right(\psi))
$\vspace{0.1cm}

\noindent\textbf{Reduction from the recurrence problem for $\DHS_L$.}
 Some ideas in the proposed reduction for the logic $\DHS_L$ are taken from~\cite{DemriLS10}, where it is shown that model checking one-counter automata against $\LTL$ with registers in undecidable.

We first provide a characterization of the recurrent computations of $M$. Let $\xi=\delta_0,\delta_1,\ldots$ be an infinite
sequence of $M$-transitions. We say that $\xi$ satisfy the \emph{consecution requirement} if (i) $\delta_0=\delta_\init$, (ii)
for all $i\geq 0$, $\delta_i$ is of the form $(q_i,\instr_i,q_{i+1})$, and (iii) for infinitely many $j\geq 0$, it holds that $\delta_j=\delta_\rec$. 
In order to characterize the sequences $\xi$ for which there exists a corresponding computation of $M$, we associate a positive natural number (called \emph{value}) to each
transition $\delta_i$ along $\xi$. For each counter $c\in \{1,2\}$, we require that the value associated to a transition $\delta_i$ of $\xi$ which increments  counter $c$ is obtained by incrementing the natural number associated to the previous incrementation of counter $c$, if any, along $\xi$. A similar requirement is imposed on the transitions along $\xi$ decrementing counter $c$ except that the values associated to $c$-decrementations must not exceed the values associated to previous $c$-incrementations. Intuitively, this ensures that at each position $i\geq 0$ along $\xi$, the value of counter $c$ is never negative. In order to simulate the zero-test, we require that for each transition $\delta_i$ associated to a zero-test for $c$, the previous values associated to $c$-incrementations correspond
to previous values associated to $c$-decrementions.

Formally, a \emph{flat configuration} is a pair $(\delta,n)$ consisting of a transition $\delta\in \Delta$ and a positive natural number $n>0$ such that $n=1$ if $\delta\in Zero(c)$ for some counter $c$. We say that $n$ is the value of $(\delta,n)$.
 A \emph{well-formed $M$-sequence} is an infinite sequence $\rho=(\delta_0,n_0),(\delta_1,n_1),\ldots$ of flat configurations satisfying the following requirements:
\begin{compactitem}
  \item The infinite sequence of transitions $\delta_0,\delta_1,\ldots$ satisfies the consecution requirement.
  \item \emph{Increment progression} (resp., \emph{Decrement progression}): for each counter $c\in \{1,2\}$, let $\xi= (\delta_{i_0},n_{i_0}),(\delta_{i_1},n_{i_1}),\ldots$ be the (possibly empty) ordered sub-sequence of the flat configurations in $\rho$ associated with  incrementation (resp., decrementation) of counter $c$. Then, $n_{i_0}=1$ and $n_{i_h}=n_{i_{h-1}}+1$ for all $0<h<|\xi|$.
  \item \emph{Increment domination}: for each $c\in \{1,2\}$ and $j\geq 0$ such that $ \delta_j\in \Dec(c)$, there is $0\leq h<j$
  such that $ \delta_h \in \Inc(c)$ and $n_h\geq n_j$.
  \item \emph{Zero-test checking}: let $c\in \{1,2\}$ and $j\geq 0$ such that $ \delta_j\in \Zero(c)$ and there are $h<j$ such that $\delta_h$ is a $c$-incrementation or $c$-decrementation. Then, the greatest $h_{\max}$ of such $h$ is associated to a $c$-decrementation and for each $h<h_{\max}$ such that $\delta_h$ is a $c$-incrementation, it holds that $n_h\leq n_{h_{\max}}$.
\end{compactitem}


\newcounter{lemma-WellFormedMSequences}
\setcounter{lemma-WellFormedMSequences}{\value{lemma}}
\newcounter{sec-WellFormedMSequences}
\setcounter{sec-WellFormedMSequences}{\value{section}}

\begin{lemma}\label{sec:WellFormedMSequences} There is a recurrent computation of $M$ \emph{iff} there is a well-formed $M$-sequence.
\end{lemma}

\noindent \textbf{Construction of the $\DHS_L$ formula $\varphi_{L,M}$.} Let $\Prop\DefinedAs \Delta \cup \{1,\#\}$.
A flat configuration $(\delta,n)$ is encoded by the finite word $\{ \delta\}\cdot \{1\}^{n}\cdot \{\#\}$.
A well-formed $M$-sequence $\rho=(\delta_0,n_0),(\delta_1,n_1),\ldots$ is encoded by the trace obtained by concatenating the
codes of the flat configurations visited by $\rho$ starting from the first one. 

We construct a $\DHS_L$ formula $\varphi_{L,M}$ characterizing the well-formed $M$-sequences.\vspace{0.1cm}

$
\varphi_{L,M} \DefinedAs \varphi_{\con}\wedge \varphi_{\inc} \wedge \varphi_{\dec} \wedge \varphi_{\zero} \wedge \varphi_{\dom}
$\vspace{0.1cm}

The conjunct $\varphi_{\con}$ is an $\AB$ formula capturing the traces which are concatenations of codes of flat configurations and satisfy 
the consecution requirement. The construction of $\varphi_{\con}$ is an easy task 
and we omit the details here. The conjunct $\varphi_{\inc}$ (resp., $\varphi_{\dec}$) ensures the increment (resp., decrement) progression requirement. We focus on the formula $\varphi_{\inc}$ (the definition of $\varphi_{\dec}$ being similar) which 
requires that (i) 
the value associated to the first $c$-incrementation, if any, is $1$, and (ii) if a $c$-incrementation $\mathcal{I}$ with value $n_1$ is followed by a $c$-incrementation
with value $n_2$, then $n_2>n_1$ and there is also a $c$-incrementation following $\mathcal{I}$ with value $n_1+1$.
The first requirement can be easily expressed by an $\AB$ formula. The second requirement is captured   by the following $\DHS_L$
formula.\vspace{0.1cm}
 
\noindent $
 \begin{array}{l}
\hspace{-0.2cm}\displaystyle{\bigwedge_{c\in\{1,2\}} \bigwedge_{\delta\in \Inc(c)}} \hsUA\hsUA \Bigl((\Left(\delta)\wedge \Right(\#)\wedge \neg \Inter(\#)) \rightarrow  \Bigl[ \neg \displaystyle{\bigvee_{\delta'\in \Inc(c)}} \hsL_{\Delta\leq 0} (\Left(\delta')\wedge \Right(\#))\,\wedge \vspace{0.1cm} \\
\bigl(  \displaystyle{\bigvee_{\delta'\in \Inc(c)}}\hsA \Right(\delta') \rightarrow \displaystyle{\bigvee_{\delta'\in \Inc(c)}}\hsL_{\Delta= 0}(\Left(\delta')\wedge \neg \Inter(\#)\wedge \RNext(\#))\bigr)\Bigr]\Bigr)
\end{array}
$\vspace{0.1cm}
 
\noindent The conjunct $\varphi_{\zero}$ expresses the zero-test checking requirement. It ensures that for each counter $c$, (i) there is no
$c$-incrementation $\mathcal{I}$ s.t.~the first $c$-operation following $\mathcal{I}$ is a zero-test, and (ii)
there is no $c$-incrementation followed by a $c$-decrementation $\mathcal{D}$ with a smaller value  such that
the first $c$-operation following $\mathcal{D}$ is a zero-test.
The first requirement can be easily expressed by an $\AB$ formula. The second requirement is captured in $\DHS_L$ as follows.\vspace{0.1cm} 
 
\noindent $
 \begin{array}{l}
\neg\displaystyle{\bigvee_{c\in\{1,2\}}\,\bigvee_{\delta_i\in \Inc(c)}\,\bigvee_{\delta_d\in \Dec(c)}\,\bigvee_{\delta_0\in \Zero(c)}} \hsA\hsA \Bigl((\Left(\delta_i)\wedge \Right(\#)\wedge \neg \Inter(\#)) \,\wedge   \vspace{0.1cm} \\
\hsL_{\Delta< 0}[\Left(\delta_d)\wedge \Right(\#)\wedge \hsA(\Right(\delta_0) \wedge \displaystyle{\bigwedge_{\delta \in \Inc(c)\cup\Dec(c)\cup\Zero(c)}}\neg \Inter(\delta))]\Bigr)
\end{array}
$\vspace{0.1cm}
 
\noindent Finally, the conjunct $\varphi_\dom$ characterizes the increment domination requirement.
One can easily check that the following conditions capture increment domination.
\begin{compactitem}
 \item If there is some $c$-decrementation, then there is some $c$-incrementation.
  \item  $c$-incrementations have values greater than previous $c$-decrementations.
  \item If a $c$-incrementation $\mathcal{I}$ with value $n$ is not followed by other $c$-incrementations, then each $c$-decrementation  following $\mathcal{I}$ has a value smaller or equal to $n$.
\end{compactitem}
We focus on the third requirement which can be expressed in $\DHS_L$ as follows (the specification of the first and second requirements are simpler):\vspace{0.1cm}
%
%
 
\noindent$
 \begin{array}{l}
 \displaystyle{\bigwedge_{c\in\{1,2\}}\,\bigwedge_{\delta_i\in \Inc(c)} } \hsUA\hsUA \Bigl(\bigl[\Left(\delta_i)\wedge \Right(\#)\wedge \neg \Inter(\#) \wedge \hsUA \displaystyle{\bigwedge_{\delta\in \Inc(c)}}\neg \Right(\delta)\bigr] \, \longrightarrow \vspace{0.1cm}\\
\neg \hsL_{\Delta>0}\displaystyle{\bigvee_{\delta_d\in \Dec(c)} } [\Left(\delta_d)\wedge \Right(\#)\wedge \neg \Inter(\#)]\Bigr)
\end{array}
$\vspace{0.1cm}

\noindent Note that the unique constant used in the constraints of $\varphi_{L,M}$ is $0$. By construction, the $\DHS_L$ formula $\varphi_{L,M}$ captures the traces encoding the well-formed $M$-sequences. Thus,
by Lemma~\ref{sec:WellFormedMSequences}, $\varphi_{L,M}$ is satisfiable iff $M$ has a recurrent computation. 

%% file: DecisionProceduresDHS.tex
\section{Decidable  fragments of \DHS}\label{DecisionProcedures}

In this section,
we show that model checking and satisfiability of $\DHSS$ are decidable though at least \TWOEXPSPACE-hard. Moreover,
by exploiting new results on the  \emph{linear-time hybrid logic} $\HL$~\cite{FRS03,SW07,BozzelliL10}, we show that
for the fragment of $\DHSS$ given by monotonic $\DHSSF{\ABBbar}$, the considered problems are exactly \EXPSPACE-complete.  Note that
$\DHSS$ represents the maximal fragment of $\DHS$ which is not covered by the undecidability results of Section~\ref{sec:undecidabilityDHS}. Additionally, we provide a characterization of $\HS$ in terms of a novel hybrid logic which lies between the one-variable and the two-variable fragment of
$\HL$. We establish that there are linear time translations from $\HS$ formulas  into equivalent formulas of the novel logic, and vice versa. This result is of independent interest since while for the one-variable fragment of $\HL$, model checking and satisfiability are \EXPSPACE-complete~\cite{SW07,BozzelliL10}, for the two-variable fragments of $\HL$, these problems are already non-elementarily
decidable~\cite{SW07,BozzelliL10}.  \vspace{0.1cm}

 \noindent \textbf{Constrained $\HL$.} $\HL$~\cite{FRS03,SW07,BozzelliL10} extends standard  $\LTL$ + past  by first-order concepts.
 Here, we consider a constrained version of $\HL$ ($\CHL$) where the temporal modalities are equipped with timing constraints. Formally,
 $\CHL$ formulas $\varphi$ over $\Prop$ and a set $X$ of (position) variables are defined by the following syntax:\vspace{0.1cm}
%

$
\varphi \DefinedAs \top \ |\ p  \ |\  x  \ |\ \neg\,\varphi \ |\ \varphi\, \wedge\, \varphi \ |\
  \Eventually_{\sim c} \varphi\ |\ \PEventually_{\sim c} \varphi  \ |\ \DBinder. \varphi
$\vspace{0.1cm}

\noindent  where $p \in \Prop$, $x\in X$, $\sim \in \{<,\leq,=,>,\geq\}$, $c\in \INT$,
$\Eventually_{\sim c}$ is the \emph{constrained strict eventually} modality and $\PEventually_{\sim c}$ is its past counterpart,   and  $\DBinder$ is the \emph{downarrow binder} operator which assigns the variable name $x$ to the current position.
A formula is \emph{monotonic} if it does not use equality constraints $= \, c$. We also exploit the constrained modalities $\Always_{\sim c}$ (\emph{always}) and $\PAlways_{\sim c}$ (\emph{past always}) as abbreviations for $\neg\Eventually_{\sim c}\neg\varphi$ and $\neg\PEventually_{\sim c}\neg\varphi$, respectively.
The standard strict eventually (resp., always) modality $\Eventually$ (resp., $\Always$) corresponds to $\Eventually_{>0}$ (resp., $\Always_{>0}$), and its past counterpart $\PEventually$ (resp., $\PAlways$) corresponds to
$\PEventually_{>0}$ (resp., $\PAlways_{>0}$). The logic $\HL$~\cite{FRS03,SW07,BozzelliL10} corresponds to the $\CHL$ fragment using only the temporal modalities $\Eventually$ and $\PEventually$.
 We denote by $\CHL_1$ and $\CHL_2$ (resp., $\HL_1$ and $\HL_2$) the one-variable and two-variable fragments of
$\CHL$ (resp., $\HL$). 
A $\CHL$ sentence is a formula where each variable $x$ is not free (i.e., occurs in the scope of  modality $\DBinder$). The size $|\varphi|$ of a $\CHL$ formula $\varphi$ is the number of distinct subformulas of $\varphi$ multiplied the number of bits for encoding the maximal constant occurring in $\varphi$.

$\CHL$ formulas $\varphi$ are interpreted over traces $w$.
For a position $i\geq 0$ and a \emph{valuation}   $g$ assigning to each variable  a position,
the satisfaction relation
$(w,i,g)\models \varphi$  is 
defined as follows (we omit the semantics of propositions and Boolean connectives which is standard):\vspace{0.1cm}

$
\begin{array}{ll}
 (w,i,g)\models x &  \Leftrightarrow i=g(x)\\
 (w,i,g)\models \Eventually_{\sim c}\,\varphi &  \Leftrightarrow \text{there is } j>i \text{ such that } j-i\sim c \text{ and } (w,j,g)\models \varphi\\
 (w,i,g)\models \PEventually_{\sim c}\,\varphi &  \Leftrightarrow \text{there is } j<i \text{ such that } i-j\sim c \text{ and } (w,j,g)\models \varphi\\
 (w,i,g)\models  \DBinder.\varphi & \Leftrightarrow
                 (w,i,g[x \mapsto i])\models \varphi
\end{array}
$\vspace{0.1cm}

\noindent where $g[x\mapsto i](x)=i$ and $g[x\mapsto i](y)=g(y)$ for $y\neq x$. 
We write $(w,i)\models \varphi$ to mean that
$(w,i,g_0)\models \varphi$, where $g_0$ maps each variable to position 0, and $w\models \varphi$ to mean that $(w,0)\models \varphi$.  \vspace{0.05cm}

\noindent \textbf{From $\DHSSF{\ABBbar}$ to $\CHL_1$.}  We show that (monotonic) $\DHSSF{\ABBbar}$ formulas can be translated in linear time into equivalent (monotonic) $\CHL_1$ sentences. For a constraint $\sim c$, we write $(\sim c)^{-1}$ for
$\sim' -c$, where $\sim'$ is $<$ (resp., $\leq$, $=$, $>$, $\geq$) if $\sim$ is $>$ (resp., $\geq$, $=$, $<$, $\leq$).

 \begin{proposition} \label{prop:FromDABtoCHLone}
 Given a (monotonic) $\DHSSF{\ABBbar}$ formula $\varphi$, one can construct in linear-time an equivalent (monotonic) $\CHL_1$ sentence.
 \end{proposition}
 \begin{proof} Fix a variable  $x$. 
 In the translation, $x$ and the current position  refer  to the left  endpoint and right endpoint  of the current interval in $\Nat$,
 respectively. 
 We can assume that the modalities for the Allen's relations $\RelB$ and $\RelBt$ occur only in a constrained form (for example, $\hsB$ corresponds to $\hsB_{<0}$).
   Formally,  the translation $f: \DHSSF{\ABBbar} \mapsto \CHL_1$ is homomorphic w.r.t.~the Boolean connectives and is inductively defined as follows:\vspace{0.1cm}

$
\begin{array}{rlrl}
\hspace{-0.3cm} f(p) & \DefinedAs  p\wedge \neg\PEventually( \neg p \wedge (x\vee \PEventually x)) \quad\quad \quad \,\,\, f(\hsA \varphi)  &\DefinedAs & \DBinder.\, (f(\varphi)\vee \Eventually f(\varphi))\\
\hspace{-0.3cm} f(\hsB_{\sim c}\varphi) &    \DefinedAs    \PEventually_{(\sim c)^{-1}} (f(\varphi) \wedge (x\vee \PEventually x))\quad\quad
f(\hsBt_{\sim c}\varphi)  &   \DefinedAs &  \Eventually_{ \sim c}  f(\varphi)
\end{array}
$\vspace{0.1cm}

 By a straightforward induction on $\varphi$, we obtain that given a trace $w$, an interval $[i,j]$, a valuation $g$ such that
 $g(x)=i$, it holds that $[i,j]\models_w \varphi$ if and only if $(w,j,g)\models f(\varphi)$. The desired
 $\CHL_1$ sentence $\varphi'$ equivalent to $\varphi$ is then defined as follows: $\varphi'\DefinedAs \DBinder. \,f(\varphi)$.
 \end{proof}\vspace{-0.1cm}

In Section~\ref{sec:MonotonicCHLOne}, we show that model checking and satisfiability of monotonic $\CHL_1$ are \EXPSPACE-complete.  By~\cite{BozMPS2021}, for the logic $\AB$ over traces, the considered problems are already \EXPSPACE-hard. 
 Thus, by Proposition~\ref{prop:FromDABtoCHLone} we obtain the following result.

 \begin{theorem}MC and satisfiability of monotonic $\DHSSF{\ABBbar}$  are   \EXPSPACE-complete.
 \end{theorem}

 \noindent \textbf{Decidability of $\DHSS$.} We first introduce a variant of $\CHL$, we call \emph{swap $\CHL$} ($\SCHL$). 
  $\SCHL$ formulas $\varphi$  are defined as follows:
$
\varphi \DefinedAs \top \ |\ p  \ |\  x  \ |\ \neg\,\varphi \ |\ \varphi\, \wedge\, \varphi \ |\
  \Eventually_{\sim c} \, \varphi\ |\ \PEventually_{\sim c} \, \varphi  \ |\ \Swap_x. \varphi.
$

\noindent  The novel modality $\Swap_x$ simultaneously
assigns to $x$ the value of the current position and updates the current position to the value previously referenced by $x$. Formally, its semantics is defined as follows:
$
 (w,i,g)\models \Swap_x . \varphi    \Leftrightarrow
                 (w,g(x),g[x \mapsto i])\models \varphi.
$
%

We are interested in the one-variable fragment $\SCHL_1$ of $\SCHL$, and in the unconstrained version
 $\SHL_1$ of $\SCHL_1$ where the unique temporal modalities are $\Eventually$ and $\PEventually$.  From a succinctness point of view, the fragment $\SCHL_1$  lies between $\CHL_1$ and $\CHL_2$ (a proof is in Appendix~\ref{APP:ComparisonHLandSHL}).

\newcounter{prop-ComparisonHLandSHL}
\setcounter{prop-ComparisonHLandSHL}{\value{proposition}}
\newcounter{sec-ComparisonHLandSHL}
\setcounter{sec-ComparisonHLandSHL}{\value{section}}

\begin{proposition}\label{prop:ComparisonHLandSHL} Given a $\CHL_1$ (resp., $\HL_1$) sentence, one can construct in linear time an equivalent $\SCHL_1$ (resp., $\SHL_1$) sentence. Moreover, given a $\SCHL_1$ (resp., $\SHL_1$) sentence, one can construct in linear time an equivalent $\CHL_2$ (resp., $\HL_2$) sentence.
\end{proposition}

 We can show that
$\DHSS$ formulas can be converted in exponential time into equivalent $\SCHL_1$ sentences.
Moreover, the logic $\HS$ over traces exactly corresponds to $\SHL_1$, i.e.,
there are linear-time translations from  $\HS$ formulas into equivalent $\SHL_1$ sentences, and vice versa.

\newcounter{prop-FromSimpleDHStoSCHL}
\setcounter{prop-FromSimpleDHStoSCHL}{\value{proposition}}
\newcounter{sec-FromSimpleDHStoSCHL}
\setcounter{sec-FromSimpleDHStoSCHL}{\value{section}}

 \begin{proposition}\label{prop:FromSimpleDHStoSCHL}
 \begin{enumerate}
   \item Given a $\DHSS$ formula $\varphi$, one can construct in singly exponential time
 an equivalent $\SCHL_1$ sentence $\psi$.  Moreover, if $\varphi$ is a $\DHSF{\BEBbarEbar}$ formula, then $\psi$ can be constructed in linear time, and $\psi\in \SHL_1$ if $\varphi\in \HS$.
   \item Given a $\SHL_1$ sentence $\varphi$, one can construct in linear time
 an equivalent $\HS$ formula $\psi$.
 \end{enumerate}
\end{proposition}

A proof of Proposition~\ref{prop:FromSimpleDHStoSCHL} is  in Appendix~\ref{APP:FromSimpleDHStoSCHL}.
By~\cite{SW07,BozzelliL10},  model checking and satisfiability of $\CHL_2$ are decidable  though with a non-elementary complexity.
We can show that for the logic $\DHSS$, the considered problems are at least \TWOEXPSPACE-hard even for the fragment given by monotonic $\DHSSF{\ABE}$ (for a proof, see Appendix~\ref{APP:LowerBoundeSimpleDHS}). Moreover, note that
 $\CHL$ formulas can be trivially translated into equivalent formulas of first-order logic $\FO$ over traces. Thus, by the first-order expressiveness completeness of the fragment $\AB$ of $\HS$~\cite{BozzelliMMPS19} (under the considered trace-based semantics), and Propositions~\ref{prop:ComparisonHLandSHL}--\ref{prop:FromSimpleDHStoSCHL}, we obtain the following result.

\begin{theorem} Model checking and satisfiability of $\DHSS$ are decidable and at least \TWOEXPSPACE-hard  even for the fragment given by monotonic $\DHSSF{\ABE}$. Moreover, $\DHSS$, monotonic $\DHSSF{\ABBbar}$, and $\HS$ have the same expressiveness.
\end{theorem}  

\subsection{\EXPSPACE-completeness of monotonic $\CHL_1$}\label{sec:MonotonicCHLOne}

In this section, we describe an asymptotically optimal automata-theoretic approach to solve satisfiability
and model checking of monotonic $\CHL_1$ ($\MCHL_1$ for short), which is based on a direct translation of $\MCHL_1$ sentences into
\emph{two-way alternating finite-state word automata} ($\TwoAWA$) equipped with standard generalized B\"{u}chi acceptance conditions. 

A $\MCHL_1$ formula is in \emph{monotonic normal form} (\MNF) if negation is applied only to atomic propositions and variables, and
the constrained temporal modalities are of the form $\OP_{\leq c}$ with $c\geq 1$ and $\OP\in \{\Eventually,\Always,\PEventually,\PAlways\}$. A $\MCHL_1$ formula $\varphi$ can be easily converted in linear-time into an equivalent $\MCHL_1$ formula in $\MNF$ $\varphi_M$, called the $\MNF$ of $\varphi$ (for details, see Appendix~\ref{APP:MonotonicNormalForm}). The \emph{dual $\widetilde{\varphi_M}$ of $\varphi_M$}
is the $\MNF$ of $\neg\varphi_M$.
\vspace{0.1cm}

\noindent \textbf{Characterization of the satisfaction relation.} We fix a monotonic $\CHL_1$ formula $\varphi$ with variable $x$, where $x$ may occur free. W.l.o.g.~we assume that $\varphi$ is in $\MNF$, and $\Prop$ is the set of atomic propositions occurring in $\varphi$.
First,  we give an operational characterization of the satisfaction relation $w\models \varphi$ which non-trivially generalizes the classical notion of Hintikka-sequence of \LTL.  Essentially, for each  trace $w$ and valuation $g$ of variable $x$, we associate to $w$  and $g$ infinite sequences  $\rho=A_0,A_1,\ldots$  of sets, where for each $i\geq 0$, $A_i$ is an \emph{atom} and intuitively describes a maximal set of  subformulas of  $\varphi$ which can hold at position $i$ along $w$ w.r.t.~the valuation $g$.
As for \LTL, the notion of atom syntactically captures the semantics of Boolean connectives. The fixpoint characterization of the unconstrained temporal modalities and the semantics of the constrained temporal modalities are locally captured by  requiring that consecutive pairs $A_i,A_{i+1}$ along the sequence $\rho$ satisfy determined syntactical constraints. Finally,  the sequence $\rho$ has to satisfy   additional non-local conditions reflecting the liveness requirements $\psi$ in the eventually subformulas $\Eventually\psi$ of $\varphi$, and the semantics of the binder modality $\DBinder$.  Now, we give the technical details.

A formula $\psi$ is  a \emph{first-level subformula} of $\varphi$ if there is an occurrence of $\psi$ in
$\varphi$ which is not in the scope of  modality $\DBinder$. The \emph{closure} $cl(\varphi)$ of $\varphi$
is the smallest set containing (i) $x$, $\top$, the propositions in $\Prop$, formula $\PEventually_{\leq 1}\top$, and (ii)
all the first-level subformulas $\psi$ of $\varphi$ together with $\Eventually_{\leq 1}\psi$ and $\PEventually_{\leq 1}\psi$, and (iii) the duals of the formulas in the points (i) and (ii). Note that $\varphi\in cl(\varphi)$ and $|cl(\varphi)|=O(|\varphi|)$.
Moreover, the set $obl(\varphi)$ of \emph{$\varphi$-obligations}  is the set of pairs of the form
$(\OP_{\leq c}\psi,d)$ such that $\OP\in\{\Eventually,\PEventually,\Always,\PAlways\}$, $\OP_{\leq c}\psi\in cl(\varphi)$, $c>1$,
and $1\leq d\leq c-1$. Intuitively, the obligations are exploited for capturing in a succinct way the semantics of the constrained temporal modalities. In particular, an obligation of the form $(\Eventually_{\leq c}\psi,d)$ asserted at a position $i$ means that there is
$j>i$ such that $\psi$ holds at position $j$, and \emph{$i+d$ is the smallest of such $j$}. Note that two distinct obligations associated
to the same formula $\Eventually_{\leq c}\psi$ cannot hold simultaneously at the same position. Dually, an
obligation of the form  $(\Always_{\leq c}\psi,d)$ asserted at a position $i$ means that there is $j>i$ such that
$\psi$ holds at all positions in $[i+1,j]$, and \emph{$i+d$ is the greatest of such $j$}. The meaning of the obligations associated to the past constrained modalities is similar. Evidently, $|obl(\varphi)|=2^{O(|\varphi|)}$.
A \emph{$\varphi$-atom} $A$ is a subset of $cl(\varphi)\cup obl(\varphi)$ such that $\top\in A$ and the following holds, where $c>1$ and $\OP\in\{\Eventually,\PEventually,\Always,\PAlways\}$:
\begin{compactitem}
  \item for each $\psi\in cl(\varphi)$, $\psi\in A$ iff $\widetilde{\psi}\notin A$;
  \item for each $\psi_1\wedge \psi_2\in cl(\varphi)$, $\psi_1\wedge \psi_2\in A$ iff $\{\psi_1,\psi_2\}\subseteq A$;
  \item for each $\psi_1\vee \psi_2\in cl(\varphi)$, $\psi_1\vee \psi_2\in A$ iff $\{\psi_1,\psi_2\}\cap A\neq \emptyset$;
  \item for each $\OP_{\leq c}\psi\in cl(\varphi)$, 
  \emph{there is at most one obligation} of the form $(\OP_{\leq c}\psi,d)$ in $A$.
\end{compactitem}
We crucially observe that 
the set $\Atoms(\varphi)$ of $\varphi$-atoms has a cardinality which is at most singly exponential in $|\varphi|$, i.e.~$|\Atoms(\varphi)|=2^{O(|\varphi|)}$. A $\varphi$-atom $A$ is \emph{initial} if $A$ does not contain formulas of the form
$\PEventually\psi$ or $\PEventually_{\leq c}\psi$  and obligations of the form $(\OP_{\leq c}\psi,d)$ with $\OP\in\{\PEventually,\PAlways\}$.

We now define  the function $\Succ_\varphi$ which maps each atom $A\in \Atoms(\varphi)$ to a subset of $\Atoms(\varphi)$.
 Intuitively, if $A$ is the atom associated with a position $i$ of the given trace $w$, then $\Succ_\varphi(A)$ contains the set of atoms associable to the next  position $i+1$ (w.r.t.~a given valuation of variable  $x$).
Formally,
$A'\in\Succ_\varphi(A)$ iff $A'$ is not initial and the following holds:\vspace{-0.2cm}

\begin{itemize}
\item \emph{$\Eventually$-requirements}: for all $\Eventually\psi\in cl(\varphi)$, $\Eventually\psi\in A$ $\Leftrightarrow$
$\{\Eventually\psi,\psi\}\cap A'\neq \emptyset$.
\item \emph{$\PEventually$-requirements}: for all $\PEventually\psi\in cl(\varphi)$, $\PEventually\psi\in A'$ $\Leftrightarrow$
$\{\PEventually\psi,\psi\}\cap A\neq \emptyset$.
\item \emph{$\Always$-Requirements}: for all $\Always\psi\in cl(\varphi)$, $\Always\psi\in A$ $\Leftrightarrow$
$\{\Always\psi,\psi\}\subseteq A'$.
\item \emph{$\PAlways$-requirements}: for all $\PAlways\psi\in cl(\varphi)$, $\PAlways\psi\in A'$ $\Leftrightarrow$
$\{\PAlways\psi,\psi\}\subseteq A$.
\item \emph{$\Eventually_{\leq c}$-requirements}: for all $\Eventually_{\leq c}\psi\in cl(\varphi)$,
    \begin{compactitem}
      \item $\Eventually_{\leq c}\psi\in A$ $\Leftrightarrow$ \emph{either}  $\psi\in A'$ \emph{or} $c>1$ and
    $(\Eventually_{\leq c}\psi,d)\in A'$ for some $1\leq d<c$;
      \item for each $1\leq d<c$,
    $(\Eventually_{\leq c}\psi,d)\in A$ $\Leftrightarrow$ \emph{either} $d=1$ and $\psi\in A'$, \emph{or} $d>1$, $\psi\notin A'$, and $(\Eventually_{\leq c}\psi,d-1)\in A'$.
    \end{compactitem}
\item \emph{$\PEventually_{\leq c}$-requirements}: for all $\PEventually_{\leq c}\psi\in cl(\varphi)$,
    \begin{compactitem}
      \item $\PEventually_{\leq c}\psi\in A'$ $\Leftrightarrow$ \emph{either}  $\psi\in A$ \emph{or} $c>1$ and
    $(\PEventually_{\leq c}\psi,d)\in A$ for some $1\leq d<c$;
      \item for each $1\leq d< c$,
    $(\PEventually_{\leq c}\psi,d)\in A'$ $\Leftrightarrow$ \emph{either} $d=1$ and $\psi\in A$, \emph{or} $d>1$, $\psi\notin A$, and $(\PEventually_{\leq c}\psi,d-1)\in A$.
    \end{compactitem}
  \item \emph{$\Always_{\leq c}$-requirements}: for all $\Always_{\leq c}\psi\in cl(\varphi)$,
    \begin{compactitem}
      \item $\Always_{\leq c}\psi\in A$ $\Leftrightarrow$ $\psi\in A'$ and, in case $c>1$, either $\Always_{\leq c}\psi\in A'$
      or  $(\Always_{\leq c}\psi,c-1)\in A'$;
      \item for each $1\leq d<c$,
    $(\Always_{\leq c}\psi,d)\in A$ $\Leftrightarrow$ \emph{either} $d=1$, $\psi\in A'$, and $\Eventually_{\leq 1}\psi\notin A'$, \emph{or}  $d>1$, $\psi \in A'$, and
    $(\Always_{\leq c}\psi,d-1)\in A'$.
    \end{compactitem}
     \item \emph{$\PAlways_{\leq c}$-requirements}: for all $\PAlways_{\leq c}\psi\in cl(\varphi)$,
    \begin{compactitem}
      \item $\PAlways_{\leq c}\psi\in A'$ $\Leftrightarrow$ $\psi\in A$ and, in case $c>1$, either    $\PAlways_{\leq c}\psi\in A$  or  $(\PAlways_{\leq c}\psi,c-1)\in A$;
      \item for each $1\leq d<c$,
    $(\PAlways_{\leq c}\psi,d)\in A'$ $\Leftrightarrow$ \emph{either} $d=1$,  $\psi\in A$,   and $\PEventually_{\leq 1}\psi\notin A$, \emph{or}  $d>1$, $\psi\in A$, and
    $(\PAlways_{\leq c}\psi,d-1)\in A$.
    \end{compactitem}
\end{itemize}\vspace{-0.1cm}

Note that 
$\Succ_\varphi$ captures the semantics of the constrained  modalities in accordance to the intended meaning of the associated obligations. 
Let $w$ be a trace and $\ell\geq 0$. A \emph{$\varphi$-sequence over the pointed trace $(w,\ell)$} is an infinite sequence $\rho=A_0,A_1,\ldots$ of $\varphi$-atoms such that:
\begin{compactitem}
  \item $A_0$ is initial, $x\in A_\ell$, and $x\notin A_i$ for each $i\neq \ell$;
  \item for each $i\geq 0$, $A_i\cap \Prop =w(i)$ (\emph{propositional consistency}), and $A_{i+1}\in\Succ_\varphi(A_i)$;
  \item \emph{Fairness}: for each $\Eventually\psi\in cl(\varphi)$ and for infinitely many $i\geq 0$, either $\psi\in A_i$ or $\Eventually\psi\notin A_i$.
\end{compactitem}
The standard fairness requirement ensures that the  
requirements $\psi$ in the first-level subformulas $\Eventually\psi$
of $\varphi$  are eventually satisfied.
In order to capture the semantics of the  modality $\DBinder$, we now give  the notion of \emph{fulfilling $\varphi$-sequence} by induction on the nesting depth of $\DBinder$. 
Formally, a $\varphi$-sequence $\rho=A_0,A_1,\ldots$ over  $(w,\ell)$ is \emph{fulfilling} if
for all $i\geq 0$ and $\DBinder.\,\psi\in A_i$, there is a fulfilling $\psi$-sequence $\rho'=A'_0,A'_1,\ldots$ over the pointed trace $(w,i)$ such that $\psi\in A'_i$.

The notion of fulfilling $\varphi$-sequence over a pointed trace $(w,\ell)$ provides a characterization
of the satisfaction relation $(w,i,g)\models \varphi$ with $g(x)=\ell$ (a proof is in Appendix~\ref{APP:characterizationCHLOne}). 

\begin{theorem}\label{theo:characterizationCHLOne} Let $\phi$ be a $\MCHL_1$ sentence in $\MNF$. Then, $w\in \Lang(\phi)$
if and only if there exists a fulfilling $\phi$-sequence $\rho=A_0,A_1,\ldots$ over $(w,0)$ such that $\phi\in A_0$.
\end{theorem}

\noindent \textbf{Automata-theoretic approach for  $\MCHL_1$.}
By Theorem~\ref{theo:characterizationCHLOne}, given a   $\MCHL_1$ sentence $\varphi$ in $\MNF$, it is not a difficult
task to construct in singly exponential time a generalized B\"{u}chi $\TwoAWA$ $\Au_\varphi$ accepting $\Lang(\varphi)$.
Given an input trace $w$, $\Au_\varphi$ guesses a $\varphi$-sequence $\rho=A_0,A_1,\ldots$ over $(w,0)$ by simulating it in forward mode along the `main' path of the run-tree.
At the $i^{th}$-node of such a path, $\Au_\varphi$ keeps track in its state of the $\varphi$-atom $A_i$. Moreover, in order to check that $\rho$ is fulfilling, for each binder formula $\DBinder.\,\psi\in A_i$, $\Au_\varphi$ recursively checks the existence
of a fulfilling $\psi$-sequence $\rho'=A'_0,A'_1,\ldots$ over $(w,i)$ by guessing the $\psi$-atom $A'_i$, with $\{x,\psi\}\subseteq A'_i$, and by activating two secondaries copies: the first one moves in backward mode by guessing the finite sequence $A'_{i-1},\ldots,A'_0$, and the second one moves in forward mode by guessing the infinite sequence $A'_{i+1},A'_{i+2},\ldots$.
Details about the construction of $\Au_\varphi$ can be found in Appendix~\ref{APP:automataCHLOne}.
By~\cite{Vardi98,KupfermanPV01},  generalized B\"{u}chi $\TwoAWA$ can be converted on the fly and in singly exponential time into equivalent  B\"{u}chi nondeterministic finite-state automata (B\"{u}chi $\NWA$).  Recall that non-emptiness of  B\"{u}chi $\NWA$ is \NLOGSPACE-complete, and  the standard model checking algorithm consists in checking emptiness of the B\"{u}chi $\NWA$ given by the synchronous product of the given Kripke structure with the B\"{u}chi $\NWA$ associated with the negation of the given formula. Thus, we obtain algorithms for satisfiability and model-checking of monotonic $\CHL_1$ which run in non-deterministic single exponential space. Therefore, since \EXPSPACE\  $=$ \NEXPSPACE, and for the logic $\HL_1$, the considered problems are already \EXPSPACE-hard~\cite{SW07}, we obtain the following result.

\begin{corollary}Model checking and satisfiability of monotonic $\CHL_1$ are  \EXPSPACE-complete.
 \end{corollary}

%% file: Conclusion.tex
\vspace{-0.3cm}
\section{Conclusion}

We have investigated decidability and complexity issues for satisfiability and model checking of
a quantitative extension of $\HS$, namely $\DHS$, under the trace-based semantics. The novel logic
provides constrained versions of the $\HS$ temporal modalities which can express bounds on the difference
between the durations of the current interval and the interval selected by the modality.
A different and natural choice would have been to consider constraints on the sum of the durations. In this setting,
one can show that the logic $\HS$ extended with sum constraints is decidable under the trace-based semantics by means
of an exponential-time translation of formulas into equivalent $\HL_2$ sentences. 

%% file: Appendix.tex
\newenvironment{changemargin}{%
  \begin{list}{}{%
    \setlength{\textheight}{25cm}%
   \setlength{\topmargin}{0.5cm}
     \setlength{\voffset}{-2cm}
      \setlength{\hoffset}{-4cm}
     \setlength{\textwidth}{15cm}%
  }%
  \item[]}{\end{list}}





\begin{changemargin}

\begin{minipage}{16cm}

\newcounter{aux}
\newcounter{auxSec}

\linespread{0.98}

\makeatletter
\edef\thetheorem{\expandafter\noexpand\thesection\@thmcountersep\@thmcounter{theorem}}
\makeatother

\begin{LARGE}
  \noindent\textbf{Appendix}
\end{LARGE}

\section{Proof of Lemma~\ref{sec:WellFormedMSequences}}\label{APP:WellFormedMSequences}

\setcounter{aux}{\value{lemma}}
\setcounter{auxSec}{\value{section}}
\setcounter{section}{\value{sec-WellFormedMSequences}}
\setcounter{lemma}{\value{lemma-WellFormedMSequences}}

\begin{lemma}There is a recurrent computation of $M$ \emph{iff} there is a well-formed $M$ sequence.
\end{lemma}
\begin{proof}
For the right implication, assume that $M$ has a recurrent computation $\pi=(\delta_0,\nu_0),(\delta_1,$ $\nu_1),\ldots$.
Note that there is a unique infinite sequence $\rho$ of flat configurations of the form $\rho=(\delta_0,n_0),(\delta_1,n_1),\ldots$ satisfying the increment progression and decrement progression requirements.
Being $\pi$ a computation of $M$, the infinite sequence of transitions $\xi=\delta_0,\delta_1,\ldots$ satisfies the consecution requirement. Moreover, because $\nu_0(c)= 0$ for each $c\in \{1,2\}$, it holds that for each $i\geq 0$ and for each counter $c$, (i) the number of $c$-incrementations along the prefix $\xi[0,i]$ of $\xi$ is greater or equal to the number of $c$-decrementations along $\xi[0,i]$, and (ii) if $\delta_i$ is a $\zero$ transition for counter $c$, then the number of $c$-incrementations and the number of $c$-decrementations along $\xi[0,i]$ coincide. It easily follows that $\rho$ satisfies the increment domination requirement and the zero-test requirement too. Hence, $\rho$ is well-formed and the result follows.

For the left implication, let $\rho=(\delta_0,n_0),(\delta_1,n_1),\ldots$ be a well-formed $M$-sequence.
A \emph{generalized configuration} is a pair $(\delta,\nu)$ consisting of a transition $\delta$, and a mapping
$\nu: \{1,2\}\mapsto \INT$ assigning to each counter an arbitrary (possibly negative) integer. Let
$\pi$ be the infinite sequence of \emph{generalized configurations} of the form $\pi=(\delta_0,\nu_0),(\delta_1,\nu_1),\ldots$
such that $\nu_0(c)= 0$ for each $c\in \{1,2\}$, and for each $i\geq 0$,  $\nu_{i+1}$ is defined in terms of $\nu_i$ in accordance to the instruction $\instr=(I,c)$ associated with the transition $\delta_i$.\details{: (i) the value of counter $3-c$ does not change, (ii) the value of counter $c$ is incremented (resp., decremented) if $I=\inc$ (resp., $I=\dec$), and (iii) the value  of counter
$c$ does not change if $I=\zero$.} We show that $\pi$ is a recurrent computation of $M$.
Since $\nu_0(c)=0$ for each counter $c$, and $\rho$ is well-formed, it holds that at each position
$i\geq 0$ and for each counter $c$, (i) the number of $c$-incrementations along the prefix $\rho[0,i]$ of $\rho$ is greater or equal to the number of $c$-decrementations along $\rho[0,i]$, and (ii) if $\delta_i$ is a zero-test for counter $c$, then the number
 of $c$-incrementations and the number of $c$-decrementations along  $\rho[0,i]$ coincide. By construction, it follows that
for each $i\geq 0$, the mapping $\nu_i$ assigns to each counter a natural number, and  if $\delta_i$ is a zero-test for counter $c$,
then $\nu_i(c)=0$. Thus, since the infinite sequence of transitions $\delta_0,\delta_1,\ldots$ satisfies the consecution requirement, we conclude that $\pi$ is a recurrent computation.
\end{proof}

 \setcounter{lemma}{\value{aux}}
 \setcounter{section}{\value{auxSec}}

 \section{Reductions from the recurrence problem for $\DHS_A$ and $\DHS_O$}\label{APP:UndecisabilityDHSForAO}

For the fragments $\DHS_A$ and $\DHS_O$, we exploit  the set of propositions  
$\Prop\DefinedAs\Delta\cup\{ 1,2,\#,\$\}$.
 A finite trace is a non-empty finite word over $2^{\Prop}$. A \emph{left counter-code} is a finite trace $w$  of the
form $w= \{2\} \cdot  \emptyset^{n+1} \cdot \{1\}\cdot  \emptyset^{m+1}$ for some $c\in\{1,2\}$ and $m,n\geq 0$. The finite trace $w$ encodes the counter valuation  assigning to counter $2$ the value $n$ and to counter $1$ the value $m$.   A \emph{right counter-code} is a finite trace $w$ whose reverse is a left counter-code: $w$ encodes the counter-valuation encoded by its reverse.
An $M$-configuration $(\delta,\nu)$ is  encoded by the finite trace $w_L\cdot \{\delta,\#\} \cdot w_R$, where $w_L$  is the left counter-code of $\nu$ and $w_R$ is the right counter-code of $\nu$ (hence, $w_R$ is the reverse of $w_L$). A computation $\pi=C_1,C_2,\ldots$ of $M$ is then encoded by the trace $\{\$\}\cdot w_{c_1} \cdot \{\$\}\cdot w_{c_2} \cdot \{\$\}\ldots$, where $w_{c_h}$ is the code of the $M$-configuration $C_h$ for each $h\geq 1$.
 The following figure illustrates
the encoding of a configuration $C=(\delta,\nu)$.

\input{ConfigurationCode.tex}

We also exploit the notion of a \emph{pseudo configuration-code} which is a finite trace of the form $w_L\cdot \{\delta,\#\} \cdot w_R$ such that $\delta\in \Delta$ and $w_L$ (resp. $w_R$) is a left counter-code (resp., right counter-code). Note that here we do not require that $w_R$ is the reverse of $w_L$. A \emph{pseudo computation-code} is a trace of the form
$\{\$\}\cdot w_{1} \cdot \{\$\}\cdot w_{2} \cdot \{\$\}\ldots$ such that
\begin{compactitem}
  \item each $w_i$ with $i\geq 1$ is a pseudo-configuration code;
  \item $w_1$ encodes the initial $M$-configuration (both counters have value $0$);
  \item for each $w_i$, let $\delta_i$ be the transition associated with $w_i$. If the instruction of $\delta_i$ is      $(\dec,c)$
  (resp, $(\zero,c)$) for some counter $c$, then the value of counter $c$ in both the left and right part
  of $w_i$ is greater than zero (resp., is equal to zero);
 \item the infinite sequence of transitions $\delta_1,\delta_2,\ldots$, where $\delta_i$ is the transition associated with $w_i$ for each $i\geq 1$, satisfies
the consecution requirement defined in Section~\ref{sec:undecidabilityDHS}.
\end{compactitem}
\end{minipage}

\begin{minipage}{16cm}

\noindent For each $X\in \{A,O\}$, we construct a $\DHS_X$ formula $\varphi_{X,M}$ capturing the codes of recurrent computations of $M$:
$
\varphi_{X,M} \DefinedAs \varphi_{\con}\wedge \varphi_{\LR}^{X} \wedge \displaystyle{\bigwedge_{\delta\in \Delta}} \varphi_{\delta}^{X}.
$

The conjunct $\varphi_{\con}$ is an $\AB$ formula capturing the traces $w$  which are pseudo-computation codes.  The construction of the formula $\AB$ is an easy task and we omit the details here. The conjunct $\varphi_{\LR}^{X}$ requires that for each pseudo-configuration code $w_L\cdot \{\delta,\#\}\cdot w_R$, the right part $w_R$ corresponds to the reverse of the left part $w_L$
(hence, $w_L\cdot \{\delta,\#\}\cdot w_R$ is a configuration code). This is ensured iff (i) $w_L$ and $w_R$ have the same length, and (ii) the suffix of $w_L$ starting at the $\{1\}$-position has the same length as the prefix of $w_R$ leading to the
$\{1\}$-position. Thus, $\varphi_{\LR}^{A}$ and $\varphi_{\LR}^{O}$ are defined as follows.\vspace{-0.2cm}
\[
 \begin{array}{ll}
\varphi_{\LR}^{A}\DefinedAs & \hsUA\hsUA \bigl((\Left(2)\wedge \Right(\#)\wedge \neg \Inter(\#)) \rightarrow
\hsA_{\Delta=0}(\Right(2)\wedge \neg \Inter(2))\bigr)\,\,\,\wedge \\
& \hsUA\hsUA \bigl((\Left(1)\wedge \Right(\#)\wedge \neg \Inter(\#)) \rightarrow
\hsA_{\Delta=0}(\Right(1)\wedge \neg \Inter(1))\bigr)\vspace{-0.1cm}
\end{array}
\]
Note that the two occurrences of $\hsA_{\Delta=0}$ in the previous formula can be replaced by the modalities
 $\hsA_{\Delta \geq 0}$ and $\hsUA_{\Delta \geq 0}$, where only the non-strict constraint $\geq 0$ is exploited. For example, the subformula $\hsA_{\Delta=0}(\Right(2)\wedge \neg \Inter(2))$ can be equivalently replaced with
$\hsA_{\Delta \geq 0}(\Right(2)\wedge \neg \Inter(2))\wedge \hsUA_{\Delta \geq 0}(\Right(2)\vee \Inter(2))$, and similarly, for the subformula $\hsA_{\Delta=0}(\Right(1)\wedge \neg \Inter(1))$. The formula $\varphi_{\LR}^{O}$ is defined as follows.\vspace{-0.2cm}
 \[
 \begin{array}{l}
  \hsUA\hsUA \bigl((\Left(2)\wedge \Right(\#)\wedge \neg \Inter(\#)) \rightarrow
\hsO_{\Delta=0}(\LNext(\#)\wedge \RNext(2)\wedge \neg \Inter(2))\bigr)\,\,\,\wedge \\
  \hsUA\hsUA \bigl((\Left(1)\wedge \Right(\#)\wedge \neg \Inter(\#)) \rightarrow
\hsO_{\Delta=0}(\LNext(\#)\wedge \RNext(1)\wedge \neg \Inter(1))\bigr)\vspace{-0.1cm}
\end{array}
\]
The two occurrences of
$\hsO_{\Delta =0}$ can be replaced by the modalities $\hsO_{\Delta \geq 0}$ and $\hsUO_{\Delta \geq 0}$. For example,
the subformula $\hsO_{\Delta=0}(\LNext(\#)\wedge \RNext(2)\wedge \neg \Inter(2))$ of $\varphi_{\LR}^{O}$ can be equivalently replaced by $\hsO_{\Delta\geq 0}(\LNext(\#)\wedge \RNext(2)\wedge \neg \Inter(2))\wedge \hsUO_{\Delta\geq 0}(\LNext(\#) \rightarrow (\Right(2)\vee \RNext(2) \vee \Inter(2)))$.

\noindent The conjunct $\varphi_{\delta}^{X}$ in the definition of $\varphi_{X,M}$ ensures that  whenever the transition $\delta$ is associated to a configuration code $\rho_C= w_L \cdot \{\delta,\#\}\cdot w_R$ along the given trace, the values of the counters in the configuration code $\rho_{C'}=w'_L \cdot \{\delta',\#\}\cdot w'_R$  following $\rho_C$ are updated consistently with the instruction of $\delta$. In the construction of  $\varphi_{\delta}^{X}$, we crucially exploit the fact that the encoding of the $2$-counter value in the left part $w'_L$ of $\rho_{C'}$ is adjacent to the encoding of the $2$-counter value in the
right part $w_R$ of $\rho_C$, as illustrated in the following figure.
We focus on the case where $\delta$ is  associated to   a decrementation, i.e. $\delta\in \Dec(c)$ for some counter $c\in\{1,2\}$. The cases for increment or zero-test instructions are similar.

\input{AdjacentConfigurationCodes.tex}

Let $\rho_R$ be the suffix of $w_R$ starting at the $\{1\}$-position, $\rho'_L$ be the prefix of $w'_L$ leading to the $\{1\}$-position, $m$ (resp., $m'$) the value of counter $1$ in $w_R$ (resp., $w'_L$), and  $n$ (resp., $n'$) the value of counter $2$ in $w_R$ (resp., $w'_L$) (see the previous figure). Being $(\dec,c)$ the instruction of $\delta$,  the formula $\varphi_\con$ ensures that $m>0$ if $c=1$, and $n>0$ otherwise.
Additionally, we need to ensure that $m'=m$ and $n'=n-1$ if $c=2$, and $m'=m-1$ and $n'=n$ otherwise. The previous requirements
hold iff (i) $|w'_L|= |w_R|-1$  and (ii) $|\rho'_L|=|\rho_R|-1$ if $c=2$, and $|\rho'_L|=|\rho_R|$ otherwise. Assume that $c=2$ (the case where $c=1$ being similar). Then Requirements~(i) and~(ii) can be expressed in $\DHS_A$ as follows:\vspace{-0.2cm}
\[
 \begin{array}{l}
  \hsUA\hsUA \bigl((\Left(\delta)\wedge \Right(2)\wedge \neg \Inter(2)) \rightarrow
\bigl[  \Inter(\psi_A) \wedge\hsA_{\Delta=0}(\RNext(\#)\wedge   \neg \Inter(\#))\bigr]  \bigr)  \\
\psi_A \DefinedAs \hsA(\Left(1)\wedge \Right(2) \wedge\neg\Inter(2)\wedge \hsA_{\Delta=0}(\RNext(1)\wedge \neg \Inter(1)))\vspace{-0.1cm}
\end{array}
\]
As for the case of the $\DHS_A$ formula $\varphi_{\LR}^{A}$, the two occurrences of the
modality $\hsA_{\Delta=0}$ in the previous formula can be replaced by the modalities $\hsA_{\Delta \geq 0}$ and $\hsUA_{\Delta \geq 0}$. The  $\DHS_O$ formula expressing Requirements~(i) and~(ii) is instead defined as follows.\vspace{-0.2cm}
 \[
 \begin{array}{l}
  \hsUA\hsUA \bigl((\Left(\delta)\wedge \Right(\$)\wedge \neg \Inter(\$)) \rightarrow
\bigl[  \Inter(\psi_O) \wedge\hsO_{\Delta=0}(\LNext(\$)\wedge \Right(\#)\wedge   \neg \Inter(\#))\bigr]  \bigr)  \\
\psi_O \DefinedAs \hsA(\Left(1)\wedge \Right(\$) \wedge\neg\Inter(\$)\wedge \hsO_{\Delta=0}(\LNext(\$)\wedge \Right(1) \wedge \neg \Inter(1)))\vspace{-0.1cm}
\end{array}
\]
Again, as for the case of the $\DHS_O$ formula $\varphi_{\LR}^{O}$, the two occurrences of the
modality $\hsO_{\Delta=0}$ in the previous formula can be replaced by $\hsO_{\Delta \geq 0}$ and $\hsUO_{\Delta \geq 0}$. By construction, the    $\DHS_A$ formula $\varphi_{A,M}$ (resp., $\DHS_O$ formula $\varphi_{O,M}$) captures the traces encoding the recurrent computations of $M$. Hence,
 $\varphi_{A,M}$ (resp., $\varphi_{O,M}$) is satisfiable iff $M$ has a recurrent computation.  

\end{minipage}

\begin{minipage}{16cm}

 \section{Proof of Proposition~\ref{prop:ComparisonHLandSHL}}\label{APP:ComparisonHLandSHL}

\setcounter{aux}{\value{proposition}}
\setcounter{auxSec}{\value{section}}
\setcounter{section}{\value{sec-ComparisonHLandSHL}}
\setcounter{proposition}{\value{prop-ComparisonHLandSHL}}

\begin{proposition} Given a $\CHL_1$ (resp., $\HL_1$) sentence, one can construct in linear time an equivalent $\SCHL_1$ (resp., $\SHL_1$) sentence. Moreover, given a $\SCHL_1$ (resp., $\SHL_1$) sentence, one can construct in linear time an equivalent $\CHL_2$ (resp., $\HL_2$) sentence.
\end{proposition}
\begin{proof} The translation function $f:\CHL_1 \mapsto \SCHL_1$ from $\CHL_1$ formulas to $\SCHL_1$ formulas is homomorphic w.r.t.~proposition, variables, Boolean connectives and temporal modalities. Moreover, for a $\CHL_1$ formula $\varphi$ using  variable $x$, $f(\DBinder.\,\varphi)$
is defined as $\Swap_x.\,\Eventually\PEventually (x\wedge f(\varphi))$. 

For the second part of Proposition~\ref{prop:ComparisonHLandSHL}, let $\varphi$ be a  $\SCHL_1$ formula using variable $x$, and let $x_1$ and $x_2$ be two distinct variables. For each $h=1,2$, we  define a $\CHL_2$ formula $F(\varphi,x_h)$ using only variables $x_1$ and $x_2$ and such that only $x_h$ can occur  free in $F(\varphi,x_h)$. The mapping $F$ is homomorphic w.r.t.~propositions, Boolean connectives and temporal modalities, and is defined as follows for variable $x$ and the swap modality: $F(x,x_h)\DefinedAs x_h$ and $F(\Swap_x.\, \varphi,x_h)\DefinedAs \Binder{x_{3-h}}.\,\Eventually\PEventually (x_h\wedge F(\varphi,x_{3-h}))$. By a straightforward induction on the structure of the $\SCHL_1$ formula $\varphi$,  
 given a trace $w$, $h=1,2$, two positions $i,j\geq 0$, a valuation $g$ s.t.~$g(x)=j$, and a valuation $g'$ s.t.~$g'(x_h)=j$, it holds that $(w,i,g)\models \varphi$ iff $(w,i,g')\models F(\varphi,x_h)$.  Hence, if $\varphi$ is a $\SCHL_1$ sentence, then $\Binder{x_h}.\,F(\varphi,x_h)$ is a $\CHL_2$ sentence equivalent to $\varphi$. 
\end{proof}

 \setcounter{proposition}{\value{aux}}
 \setcounter{section}{\value{auxSec}}

\section{Proof of Proposition~\ref{prop:FromSimpleDHStoSCHL}}\label{APP:FromSimpleDHStoSCHL}

\setcounter{aux}{\value{proposition}}
\setcounter{auxSec}{\value{section}}
\setcounter{section}{\value{sec-FromSimpleDHStoSCHL}}
\setcounter{proposition}{\value{prop-FromSimpleDHStoSCHL}}

 \begin{proposition}
 \begin{enumerate}
   \item Given a $\DHSS$ formula $\varphi$, one can construct in singly exponential time
 an equivalent $\SCHL_1$ sentence $\psi$.  Moreover, if $\varphi$ is a $\DHSF{\BEBbarEbar}$ formula, then $\psi$ can be constructed in linear time, and $\psi\in \SHL_1$ if $\varphi\in \HS$.
   \item Given a $\SHL_1$ sentence $\varphi$, one can construct in linear time
 an equivalent $\HS$ formula $\psi$.
 \end{enumerate}
\end{proposition}
\setcounter{proposition}{\value{aux}}
 \setcounter{section}{\value{auxSec}}

\noindent  \textbf{Proof of Statement 1 in Proposition~\ref{prop:FromSimpleDHStoSCHL}.} First note that $\DHSS$ corresponds to $\DHSF{\BDEBbarDbarEbar}$ since the unconstrained modalities for the Allen's relations
$\RelA$, $\RelL$, and $\RelO$ and their inverses can be expressed in linear time into $\BEBbarEbar$. Moreover, the constrained versions
of the modalities $\hsD$ and $\hsDt$ can be easily expressed in $\DHSF{\BEBbarEbar}$ though with a singly exponential blow-up.
For example, for $n>0$, $\hsDt_{\geq  n}\varphi$ is equivalent to $\displaystyle{\bigvee_{n_1\geq 0,n_2\geq 0: n_1+n_2=n}} \hsBt_{\geq n_1}\hsEt_{\geq n_2}\varphi$. Thus, it suffices to show that a $\DHSF{\BEBbarEbar}$ formula can be converted in linear time into an equivalent $\SCHL_1$  sentence.
Let $x$ be a position variable. We use the expression $x< cur$ to indicate that the position referenced by variable $x$ is smaller than the current position.  The meaning of the expression $x>cur$ (resp., $x=cur$) is similar.
Given a  $\DHSF{\BEBbarEbar}$ formula $\varphi$ and $\tau\in \{x<cur,x>cur,x=cur\}$, we inductively define an $\SCHL_1$ formula $f(\varphi,\tau)$ using variable $x$. Intuitively, the position referenced by $x$ and the current position represent the endpoints of the interval on which $\varphi$ is currently evaluated.
We can assume that $\varphi$ contains only constrained temporal modalities. Indeed,  the modalities in  $\BEBbarEbar$  can be trivially converted into equivalent constrained versions. 
The mapping $f$ is homomorphic w.r.t.~the Boolean connectives and is inductively defined as follows:
  \begin{compactitem}
  \item $f(p,x<cur)\DefinedAs  p\wedge \neg \PEventually (\neg p \wedge (x\vee \PEventually x))$.
    \item $f(\hsB_{\Delta \sim c}\varphi,x<cur)\DefinedAs   \PEventually_{(\sim c)^{-1}}[(f(\varphi,x=cur)\wedge x)\vee (f(\varphi,x<cur)\wedge \PEventually x)]$.
  \item $f(\hsBt_{\Delta \sim c}\varphi,x<cur)\DefinedAs   \Eventually_{\sim c}\, f(\varphi,x<cur)$.
   \item $f(\hsE_{\Delta \sim c}\varphi,x<cur)\DefinedAs  \Swap_x. \, \Eventually_{(\sim c)^{-1}}[(f(\varphi,x=cur)\wedge x)\vee (f(\varphi,x>cur)\wedge \Eventually x)]$.
    \item $f(\hsEt_{\Delta \sim c}\varphi,x<cur)\DefinedAs \Swap_x. \,   \PEventually_{\sim c}\, f(\varphi,x>cur)$.
  \item $f(p,x>cur)\DefinedAs  p\wedge \neg \Eventually (\neg p \wedge (x\vee \Eventually x))$.
    \item $f(\hsB_{\Delta \sim c}\varphi,x>cur)\DefinedAs \Swap_x. \,  \PEventually_{(\sim c)^{-1}}[(f(\varphi,x=cur)\wedge x)\vee (f(\varphi,x<cur)\wedge \PEventually x)]$.
  \item $f(\hsBt_{\Delta \sim c}\varphi,x>cur)\DefinedAs \Swap_x. \,  \Eventually_{\sim c}\, f(\varphi,x<cur)$.
   \item $f(\hsE_{\Delta \sim c}\varphi,x>cur)\DefinedAs   \Eventually_{(\sim c)^{-1}}[(f(\varphi,x=cur)\wedge x)\vee (f(\varphi,x>cur)\wedge \Eventually x)]$.
    \item $f(\hsEt_{\Delta \sim c}\varphi,x>cur )\DefinedAs  \PEventually_{\sim c}\, f(\varphi,x>cur)$.
     \item $f(p,x=cur)\DefinedAs  p$.
 \item $f(\hsX_{\Delta \sim c}\varphi,x=cur)\DefinedAs  \neg \top$ for each $X\in\{B,E\}$.
  \item $f(\hsBt_{\Delta \sim c}\varphi,x=cur)\DefinedAs   \Eventually_{\sim c}\, f(\varphi,x<cur)$.
\item $f(\hsEt_{\Delta \sim c}\varphi,x=cur)\DefinedAs   \PEventually_{\sim c}\, f(\varphi,x>cur)$.
 \end{compactitem}
 \end{minipage}

\begin{minipage}{16cm}
 By a straightforward induction on the structure of $\varphi$, we obtain that given a trace $w$, a position $j\geq 0$, and a valuation $g$ such that $g(x)=j$, the following holds:
 \begin{compactitem}
   \item $[j,j]\models_w \varphi$ iff $(w,j,g)\models f(\varphi,x=cur)$.
   \item For each $i>j$, $[j,i]\models_w \varphi$ iff $(w,i,g)\models f(\varphi,x<cur)$.
   \item For each $i<j$, $[i,j]\models_w \varphi$ iff $(w,i,g)\models f(\varphi,x>cur)$.
 \end{compactitem}

 It follows that the $\SCHL_1$ sentence $\Swap_x.\,f(\varphi,x=cur )$ is equivalent to $\varphi$. Note that the previous sentence is in $\SHL_1$ if $\varphi\in\HS$. \qed \vspace{0.1cm}

 \noindent \textbf{Proof of Statement 2 in Proposition~\ref{prop:FromSimpleDHStoSCHL}.}
Given a  $\SHL_1$ formula $\varphi$ using variable $x$  and $\tau\in \{x<cur,x>cur,x=cur\}$, we 
define an $\HS$ formula $f(\varphi,\tau)$. Intuitively, the position referenced by $x$ and the current position in the valuation of $\varphi$ represent the endpoints of the interval on which $f(\varphi,\tau)$ is currently evaluated. The mapping $f$, homomorphic w.r.t.~the Boolean connectives, is inductively defined as follows:
  \begin{compactitem}
  \item $f(x,\tau)\DefinedAs  \neg \top$ for each $\tau\in \{x<cur ,x>cur \}$.
    \item $f(x,x=cur)\DefinedAs  \top$.
  \item $f(p,x<cur)\DefinedAs  \hsE (\len{1}\wedge p)$.
  \item $f(p,x>cur)\DefinedAs  \hsB (\len{1}\wedge p)$.
  \item $f(p,x=cur)\DefinedAs  p$.
  \item $f(\Swap_x.\varphi,x<cur) \DefinedAs f(\varphi,x>cur)$.
  \item $f(\Swap_x.\varphi,x>cur) \DefinedAs f(\varphi,x<cur)$.
  \item $f(\Swap_x.\varphi,x=cur) \DefinedAs f(\varphi,x=cur)$.
  \item $f(\Eventually \varphi,\tau) \DefinedAs \hsBt f(\varphi,x<cur )$ for each $\tau\in \{x<cur ,x=cur \}$.
    \item $f(\Eventually \varphi,x>cur ) \DefinedAs \hsE (\neg\len{1}\wedge f(\varphi,x>cur )) \vee
     \hsE (\len{1}\wedge f(\varphi,x=cur )) \vee \hsA (\neg\len{1}\wedge f(\varphi,x<cur ))$.
  \item $f(\PEventually \varphi,\tau) \DefinedAs \hsEt f(\varphi,x>cur )$ for each $\tau\in \{x>cur ,x=cur \}$.
        \item $f(\PEventually \varphi,x<cur ) \DefinedAs \hsB (\neg\len{1}\wedge f(\varphi,x<cur )) \vee
     \hsB (\len{1}\wedge f(\varphi,x=cur )) \vee \hsAt (\neg\len{1}\wedge f(\varphi,x>cur ))$.
  \end{compactitem}\vspace{0.1cm}

\noindent By a straightforward induction on the structure of $\varphi$, we obtain that given a trace $w$, a position $j\geq 0$, and a valuation $g$ such that $g(x)=j$, the following holds:
 \begin{compactitem}
   \item  $(w,j,g)\models \varphi$ iff $[j,j]\models_w f(\varphi,x=cur )$.
   \item For each $i>j$, $(w,i,g)\models \varphi$ iff  $[j,i]\models_w f(\varphi,x<cur )$.
   \item For each $i<j$, $(w,i,g)\models \varphi$ iff  $[i,j]\models_w f(\varphi,x>cur )$.
 \end{compactitem}
 Hence, the $\HS$ formula $f(\varphi,x=cur )$ is equivalent to $\varphi$.  This concludes the proof of Proposition~\ref{prop:FromSimpleDHStoSCHL}.\qed

\section{\TWOEXPSPACE-hardness of $\DHSS$}\label{APP:LowerBoundeSimpleDHS}

In this section, we establish the following result.

\begin{theorem}\label{theo:LowerBoundeSimpleDHS}
Model checking and satisfiability
of $\DHSS$ are at least \TWOEXPSPACE-hard even for the fragment given by monotonic
$\DHSSF{\ABE}$.
\end{theorem}

We focus on \TWOEXPSPACE-hardness of the satisfiability problem for monotonic $\DHSSF{\ABE}$ (the result for model checking is similar). 
The proof is by a polynomial-time reduction from a domino-tiling problem for grids with rows of length
$2^{2^{n}}$~\cite{Boas97}, where $n$ is an input
parameter. Formally, an instance $\Instance$ of this problem is a
tuple $\Instance =\tpl{C,\Delta,n,d_\Init,d_\acc}$, where $C$ is a
finite set of colors, $\Delta \subseteq C^{4}$ is a set of tuples
$\tpl{c_\Down,c_\Left,c_\Up,c_\Right}$ of four colours, called
\emph{domino-types}, $n>0$ is a natural number encoded in
\emph{unary}, and $d_\Init,d_\acc\in\Delta$ are two designated domino-types.  A  \emph{grid of $\Instance$}  is a mapping
$f:\Nat\times [0,2^{2^{n}}-1] \rightarrow \Delta$. Thus a grid has an infinite number of rows where
each row consists of $2^{2^{n}}$ cells, and each cell contains a
domino type.  A \emph{tiling  of $\Instance$} is a grid
$f$ satisfying the following additional constraints:\vspace{-0.2cm}

\begin{description}
 \item [Initialization:]
$f(0,0)=d_\Init$
     \item [Row adjacency:]  two adjacent cells in a row
       have the same color on the shared edge: for all
       $(i,j)\in \Nat\times [0,2^{2^{n}}-2]$,
       $[f(i,j)]_{\Right}=[f(i,j+1)]_{\Left}$.
  \item [Column adjacency:]  two adjacent cells in a column have the same color on the shared edge: for all $(i,j)\in \Nat\times [0,2^{2^{n}}-1]$,
   $[f(i,j)]_{\Up}=[f(i+1,j)]_{\Down}$.
 \item[Acceptance:] for infinitely many $\ell\in \Nat$, $f(\ell,0)=d_\acc$.\vspace{-0.2cm}
\end{description}

The  problem of checking the existence  of a  tiling for $\Instance$ is  \TWOEXPSPACE-complete~\cite{Boas97}.
We build, in time polynomial in
 the size of 
 $\Instance$, a monotonic $\DHSSF{\ABE}$ formula $\varphi_\Instance$ s.t.~$\Instance$ has a tiling
 \emph{iff} $\varphi_\Instance$ is satisfiable.
\end{minipage}

\begin{minipage}{16cm}
 \noindent \textbf{Encoding of grids.}
 First, we define a suitable encoding of the grids of $\Instance$ by traces over the set
$\Prop$ of atomic propositions given by $\Prop\DefinedAs \Delta\cup\{\#_1,\#_2, 0,1 \}$,
 Essentially, the  code of a grid $f$ is obtained by
concatenating the codes of the rows of $f$ starting from the first row. The code of
a row is in turn obtained by concatenating the codes of the row's cells starting from the first cell.
In the encoding of a cell of a grid, we keep track of
the content of the cell  together with a suitable encoding of the cell
number which is an integer in $[0,2^{2^{n}}-1]$.
We use a $2^{n}$-bit counter for expressing integers in the range
$[0,2^{2^{n}}-1]$ and an $n$-bit counter for keeping track of the position (index) $i\in [0,2^{n}-1]$ of the $(i+1)^{th}$-bit of each valuation $v$ of the $2^{n}$-bit counter. In particular, such a valuation $v\in [0,2^{2^{n}}-1]$ is encoded by a sequence of $2^{n}$ blocks of length $n+1$ where
for each $i\in [0,2^{n}-1]$, the $(i+1)^{th}$ block encodes both the value and the index of the $(i+1)^{th}$-bit in the binary representation of $v$.

Formally, a \emph{block} is a finite word $bl$ over $2^{\Prop}$ of length $n+1$ of the form $bl =\{\#_1,bit\}\{bit_1\},\ldots,\{bit_n\}$ where $bit,bit_1,\ldots,bit_n\in \{0,1\}$. The \emph{content} of $bl$ is $bit$, and the \emph{index} of $bl$ is the number in $[0,2^{n}-1]$ whose binary code is
$bit_1,\ldots,bit_n$. A \emph{cell code} is a finite word $\nu$ of length $(n+1)*2^{n}+1$ of the
form $\nu= \{d\}\cdot bl_0\cdot \ldots \cdot bl_{2^{n}-1}$, where $d\in \Delta$, and for each $i\in [0,2^{n}-1]$, $bl_i$
is a block having index $i$. The \emph{content} of $\nu$ is the domino-type $d$, and the  \emph{index} of $\nu$ is the natural number
in $[0,2^{2^{n}}-1]$ whose binary code is $bit_0,\ldots,bit_{2^{n}-1}$, where $bit_i$ is the content of the block $bl_i$
for each $i\in [0,2^{n}-1]$. A row is then encoded by a finite word of the form $\{\#_2\}\cdot cell_0\cdot\ldots \cdot cell_{2^{2^{n}}-1}$,
where $cell_i$ is a cell code of index $i$ encoding the $(i+1)^{th}$ cell of the row for each $i\in [0,2^{2^{n}}-1]$.\vspace{0.1cm}

 \noindent \textbf{Construction of the monotonic $\DHSSF{\ABE}$ formula $\varphi_\Instance$.} We now construct in polynomial time
 a  monotonic $\DHSSF{\ABE}$ formula $\varphi_\Instance$ capturing the codes of the $\Instance$-tilings.
   Formula $\varphi_\Instance$ is defined as $\varphi_\Instance\DefinedAs \psi_{grid}\wedge  \psi_{col}$.
The conjunct $\psi_{grid}$ is an $\ABE$ formula capturing the traces $w$ which are codes of grids $f$ satisfying
the initialization, the row adjacency, and the acceptance requirement, while the conjunct $\psi_{col}$ enforces
the column adjacency requirement. The construction of $\psi_{grid}$ follows a pattern similar to the one exploited
for proving \EXPSPACE-completeness of the fragment $\BE$ of $\HS$~\cite{BMMPS16}.
Thus, we omit the details about the construction of $\psi_{grid}$, and we focus on the definition of
 $\psi_{col}$, where the use of the constrained versions of the modalities $\hsB$ and $\hsE$ is crucial.

 In order to ensure the column adjacency requirement, we need to ensure that for each cell code $cell$ along a row $r$ of the trace $w$ (encoding a grid), denoted by $cell'$ the cell of the row following $r$ along $w$ and having the same index as cell, it holds that $(d)_{\Up}=(d')_{\Down}$,
 where $d$ is the content of $cell$, and $d'$ is the content of $cell'$. In order to express this requirement,
we define auxiliary monotonic $\DHSSF{\ABE}$ formulas. Recall that proposition $\#_1$ marks the first position of a block, while
$\#_2$ marks the first position of a row.  
%
%
For each $p\in\Prop$, it is easy to define an $\AB$ formula $\psi_{one}(p)$ ensuring that proposition $p$ occurs exactly once in the current interval, and an $\AB$ formula $\psi_{cell}$ which holds at an interval $I$ of a grid code if $I$ corresponds to a cell code. Moreover, the following $\BE$ formula $\psi_{comp\_blocks}$  holds at an interval $I$ of a grid code iff $I$ starts at a block $bl$ and ends at a block $bl'$ such that $bl$ and $bl'$ have the same index.\vspace{0.1cm}
 
\hspace{0.2cm}$
  (\Left(\#_1)\wedge \hsE(\len{n+1}\wedge \Left(\#_1))) \wedge 
  \displaystyle{\bigwedge_{b\in \{0,1\}}\bigwedge_{i=2}^{n+1}}\hsB( \len{i}\wedge right(b)) \rightarrow \hsE(\len{n+2-i}\wedge \Left(b))
$\vspace{0.1cm}

\noindent Let $c$ be the integer  given by $c\DefinedAs  (n+1)*2^{n}+1$. Recall that a cell code has length $c$.
We now define a $\DHSF{\BE}$ formula $\psi_{comp\_cell}$ which holds at an interval $I$ of a grid code iff $I$ starts with a cell code
$cell$ and ends with a cell code $cell'$ such that (i) $cell$ and $cell'$ belong  to two adjacent row codes, and (ii)
$cell$ and $cell'$ have the same index. In order to satisfy the second requirement, we exploit the compound modality $\hsUB_{\Delta\geq -c}\hsUE_{\Delta\geq -c}$ for universally selecting the infixes $I'$ of $I$ which start  at a position of $cell$
and end at a position of $cell'$. Then, we require that for all those $I'$ which start with a block $bl$ and end with a block $bl'$ (note that $bl$ is a block of $cell$ and $bl'$ is a block of $cell'$) such that $bl$ and $bl'$ have the same index, it holds that $bl$ and $bl'$ have the same content as well.  \vspace{-0.2cm}
\[
\begin{array}{ll}
\psi_{comp\_cells}\DefinedAs & \hsB \psi_{cell}\wedge \hsE \psi_{cell} \wedge \psi_{one}(\#_2) \wedge \\
 & \hsUB_{\Delta\geq -c}\hsUE_{\Delta\geq -c} [\psi_{comp\_blocks} \rightarrow \displaystyle{\bigvee_{b\in\{0,1\}}}(\Left(b)\wedge \hsE (\len{n+1}\wedge \Left(b))]\vspace{-0.2cm}
\end{array}
\]
\noindent where $c=(n+1)*2^{n}+1$. Finally, the conjunct $\psi_{col}$ of $\varphi_\Instance$ ensuring the adjacency column requirement along a grid code is defined as follows:\vspace{0.1cm}

\hspace{0.2cm}$
\psi_{col}\DefinedAs \hsUA\hsUA \,[\psi_{comp\_cells}\,\, \longrightarrow\displaystyle{\bigvee_{d,d'\in \Delta:\,  (d)_{\Up}=(d')_{\Down}}}(\Left(d)\wedge \hsE (\psi_{cell}\wedge \Left(d\,')))\, ]
$\vspace{0.1cm}

By construction, there is a tiling of $\Instance$ iff $\varphi_{\Instance}$ is satisfiable. This concludes the proof of Theorem~\ref{theo:LowerBoundeSimpleDHS}.
\end{minipage}

\begin{minipage}{16cm}
\section{Monotonic normal form of $\MCHL_1$ formulas}\label{APP:MonotonicNormalForm}

Let $\varphi$ and $\psi$ be two $\MCHL_1$ formulas. We say that $\varphi$ and $\psi$ are \emph{strongly equivalent}, written
$\varphi\equiv\psi$, if for all traces $w$, positions $i$, and valuations $g$, $(w,i,g)\models \varphi \Leftrightarrow (w,i,g)\models \psi$. By using De Morgan's laws and the following strong equivalences, a $\MCHL_1$ formula can be easily converted in linear-time into an equivalent $\MCHL_1$ formula in $\MNF$: for $c\geq 1$,  (i) $\neg \DBinder.\,\psi \equiv \DBinder.\, \neg\psi$, (ii) $\Eventually_{>c}\psi\equiv \Always_{\leq c}\Eventually\,\psi$, (iii) $\PEventually_{>c}\psi\equiv \PEventually\top\wedge \PAlways_{\leq c}\PEventually\,\psi$, (iv) $\Always_{>c}\psi\equiv \Eventually_{\leq c}\Always\,\psi$, and (v)
$\PAlways_{>c}\psi\equiv \PEventually_{\leq c}\PAlways\,\psi$.

\section{Proof of Theorem~\ref{theo:characterizationCHLOne}}\label{APP:characterizationCHLOne}

Theorem~\ref{theo:characterizationCHLOne} directly follows from the following  Lemmata~\ref{lemma:completenessCHLOne}--\ref{lemma:correctnessCHLOne}, where $\varphi$ is a monotonic $\CHL_1$ formula in $\MNF$ whose unique variable $x$ may occur free.

\begin{lemma}[Completeness]\label{lemma:completenessCHLOne} Assume that  $(w,i,g)\models \varphi$. Then, there exists a
 fulfilling $\varphi$-sequence  $\rho=A_0,A_1,\ldots$ over the
pointed trace $(w,g(x))$ such that $\varphi\in A_i$.
\end{lemma}
\begin{proof} The proof is by induction on the structure of $\varphi$. For each $h\geq 0$, let $A_h$ be the subset of
$cl(\varphi)\cup obl(\varphi)$ satisfying the following requirements:
\begin{compactitem}
  \item for each $\psi\in cl(\varphi)$, $\psi\in A_h$ $\Leftrightarrow$ $(w,h,g)\models \psi$;
  \item for each $(\Eventually_{\leq c}\psi,d)\in obl(\varphi)$, $(\Eventually_{\leq c}\psi,d)\in A_h$ $\Leftrightarrow$
  there are $j>h$ such that $(w,j,g)\models \psi$ and $h+d$ is the smallest of such $j$;
  \item for each $(\PEventually_{\leq c}\psi,d)\in obl(\varphi)$, $(\PEventually_{\leq c}\psi,d)\in A_h$ $\Leftrightarrow$ $h\geq d$,
  there are $j<h$ such that $(w,j,g)\models \psi$, and $h-d$ is the greatest of such $j$;
    \item for each $(\Always_{\leq c}\psi,d)\in obl(\varphi)$, $(\Always_{\leq c}\psi,d)\in A_h$ $\Leftrightarrow$
  there are $j>h$ such that $(w,k,g)\models \psi$ for all $k\in [h+1,j]$, and $h+d$ is the greatest of such $j$;
      \item for each $(\PAlways_{\leq c}\psi,d)\in obl(\varphi)$, $(\PAlways_{\leq c}\psi,d)\in A_h$ $\Leftrightarrow$ $h\geq d$,
  there are $j<h$ such that $(w,k,g)\models \psi$ for all $k\in [j,h-1]$, and $h-d$ is the smallest of such $j$.
\end{compactitem}
By construction, it easily follows that $A_h$ is a $\varphi$-atom and the infinite sequence $\rho=A_0,A_1,\ldots$ is a $\varphi$-sequence over $(w,g(x))$ such that $\varphi\in A_i$ (recall that by hypothesis $(w,i,g)\models \varphi$). If $\varphi$ does not contain occurrences of the  modality $\DBinder$, then
$\rho$ is obviously fulfilling. Otherwise,  we apply the induction hypothesis, for concluding that $\rho$ is fulfilling in this case too.
\end{proof}

\begin{lemma}[Correctness] \label{lemma:correctnessCHLOne} Let $\rho=A_0,A_1,\ldots$ be a fulfilling $\varphi$-sequence over the
pointed trace $(w,\ell)$, and $g$ be a valuation such that $g(x)=\ell$. Then for all $i\geq 0$ and $\psi\in cl(\varphi)$, it holds that $\psi\in A_i \Leftrightarrow (w,i,g)\models \psi$.
\end{lemma}
\begin{proof}
Fix $i\geq 0$ and $\psi\in cl(\varphi)$. The proof is by induction on the structure of $\psi$. The cases where
$\psi$ is a  proposition or $\psi=x$ directly follow from the definition of $\varphi$-sequence. The cases where the root operator of $\psi$ is either a Boolean connective or an unconstrained temporal modality easily follow from the definitions of $\varphi$ and of the function $\Succ_\varphi$, the fairness requirement in the definition of $\varphi$-sequence, and the induction hypothesis. We now consider the remaining cases:
\begin{compactitem}
  \item $\psi=\Eventually_{\leq c}\theta$ with $c\geq 1$: by the $\Eventually_{\leq c}$-requirements in the definition of $\Succ_\varphi$, it easily follows that $\psi\in A_i$ $\Leftrightarrow$ there is $j>i$ such that $j-i\leq c$ and $\theta\in A_j$. Thus, by the induction hypothesis on $\theta$, the result follows.
  \item $\psi=\PEventually_{\leq c}\theta$ with $c\geq 1$: this case is similar to the previous one. 
  \item $\psi=\Always_{\leq c}\theta$ with $c\geq 1$: by the $\Always_{\leq c}$-requirements in the definition of $\Succ_\varphi$, it easily follows that $\psi\in A_i$ $\Leftrightarrow$ there is $j>i$ such that $j-i\leq c$ and $\theta\in A_h$ for all $h\in [i+1,j]$. Thus, by the induction hypothesis on $\theta$, the result follows.
  \item $\psi=\PAlways_{\leq c}\theta$ with $c\geq 1$: this case is similar to the previous. 
  \item $\psi=\DBinder.\,\theta$. First, assume that $\DBinder.\,\theta\in  A_i$. Since $\rho=A_0,A_1,\ldots$ is fulfilling, there exists a fulfilling $\theta$-sequence $\rho'=A'_0,A'_1,\ldots$ over $(w,i)$ such that $\theta\in A'_i$. By the induction hypothesis, it follows that $(w,i,g')\models \theta$, where $g'(x)=i$. Hence, $(w,i,g)\models \DBinder.\,\theta$, and the result follows. For the converse direction, assume that $(w,i,g)\models \DBinder.\,\theta$. We need to show that $\DBinder.\,\theta\in A_i$. We assume the contrary and derive a contradiction. By definition of $\varphi$-atom,
       $\DBinder.\,\widetilde{\theta}\in A_i$. Since
      $\rho=A_1,A_2,\ldots$ is a fulfilling $\varphi$-sequence over $(w,\ell)$, there must exist
      a fulfilling $\widetilde{\theta}$-sequence $\rho'=A'_0,A'_1,\ldots$ over $(w,i)$ such that
      $\widetilde{\theta}\in A'_i$. By the induction hypothesis, it follows that
      $(w,i,g')\models \widetilde{\theta}$ where $g'(x)=i$. Hence, $(w,i,g)\models \DBinder.\,\widetilde{\theta}$
      which contradicts the hypothesis $(w,i,g)\models \DBinder.\,\theta$. 
\end{compactitem}
\end{proof}
\end{minipage}

\begin{minipage}{16cm}
\section{From monotonic $\CHL_1$ to generalized B\"{u}chi \TwoAWA}\label{APP:automataCHLOne}

In this section, we formalize the translation of monotonic $\CHL_1$ sentences into generalized B\"{u}chi $\TwoAWA$ based on the result of Theorem~\ref{theo:characterizationCHLOne}. We first recall the class of generalized B\"{u}chi $\TwoAWA$.

For a  set $X$, $\Bool^{+}(X)$   denotes the set
of \emph{positive} Boolean formulas over $X$ 
(we also allow the formulas $\True$ and
$\False$). A subset $Y$ of $X$  \emph{satisfies}
$\xi\in\Bool^{+}(X)$ iff the truth assignment that assigns $\True$
to the elements in $Y$ and $\False$ to the elements of $X\setminus
Y$ satisfies $\xi$.

A tree $T$ is a   prefix closed subset of $\Nat^{*}$.  Elements of $T$ are called nodes and $\varepsilon$ (the empty word) is the root of $T$. For $x\in T$, a child of $x$ in $T$ is a node in $T$
  of the form $x\cdot n$ for some $n\in \Nat$.  A  path of $T$ is a maximal sequence $\pi$ of nodes such that   $\pi(i)$ is a child in $T$ of $\pi(i-1)$ for all $0<i<|\pi|$.   For an alphabet $\Sigma$, a $\Sigma$-labeled   tree is a pair $\tpl{T, \Lab}$, 
   where $T$ is a tree and $\Lab:T \mapsto \Sigma$. 

A generalized B\"{u}chi \TwoAWA\  is a tuple
 $\Au=\tpl{\Sigma,Q,Q_0,\delta,F_-,\Frag}$, where
$\Sigma$ is a finite input alphabet, $Q$ is a finite set of states, $Q_0\subseteq Q$ is a set of initial states,
$\delta:Q\times \Sigma\rightarrow \Bool^{+}(\{\downarrow,\uparrow\}\times Q)$ is
a transition function, $F_-\subseteq Q$ is a \emph{backward acceptance condition},   and $\Frag=\{F_1,\ldots,F_k\}$ is
a family  of sets $F_i$ of accepting states. 
Intuitively, when  $\Au$ is in state $q$, reading the 
symbol $\sigma\in\Sigma$,  then $\Au$ chooses a set of pairs direction/state  $\{(dir_1,q_1),\ldots,(dir_k,q_k)\}$ satisfying $\delta(q,\sigma)$ and
splits in $k$ copies such that the $i^{th}$ copy moves  to the next input symbol  in state $q_i$ if $dir_i= \downarrow$,
and to the previous input symbol in state $q_i$ otherwise.

Formally, a run over an infinite word $w\in \Sigma^\omega$ is  a $Q\times (\Nat\cup\{-1\})$-labeled tree $r=\tpl{T_r,\Lab_r}$   such that the root is labeled by $(q_0,0)$ for some $q_0 \in Q_0$
and for each node $x\in T_r$ with label $(q,i)\in Q\times \Nat$ (describing a copy of $\Au$ in state $q$ which reads $w(i)$),
there is a (possibly empty) set
$H=\{(dir_1,q_1),\ldots,(dir_k,q_k\}\subseteq \{\downarrow,\uparrow\}\times Q$ satisfying $\delta(q,w(i))$ such
that $x$ has $k$ children $x_1,\ldots, x_k$, and for
 $\ell\in [1,k]$, $x_\ell$ has label $(q_\ell,i+1)$ if $dir_\ell = \downarrow$, and label $(q_\ell,i-1)$ otherwise.
  The run $r$ is \emph{accepting} if (i) for each node with label $(q,-1)$, $q\in F_-$, and (ii) for each
infinite path $x_0 x_1\ldots$ in the tree and each accepting component $F\in \Frag$, there are infinitely many $i\geq 0$ such that the state of $x_i$ is  in $F$.
The $\omega$-language  $\Lang(\Au)$ of $\Au$ is the set of infinite
words $w\in\Sigma^\omega$  such that there is an accepting run
 of $\Au$ over $w$.

 By exploiting Theorem~\ref{theo:characterizationCHLOne}, we deduce the following result.

\begin{theorem}\label{theo:automataCHLOne} Given a monotonic $\CHL_1$ sentence $\varphi$,   one can build in singly exponential time a generalized B\"{u}chi $\TwoAWA$ over $2^{\Prop}$ with $2^{O(\varphi)}$ states such that
$\Lang(\Au_\varphi)=\Lang(\varphi)$.
\end{theorem}
\begin{proof}
  W.l.o.g.~we assume that $\varphi$ is in $\MNF$. Let $SubF(\varphi)$ be the set of subformulas of $\varphi$.
The generalized B\"{u}chi $\TwoAWA$ $\Au_\varphi=\tpl{2^{\Prop},Q,Q_0,\delta,F_-,\Frag}$ is defined as follows:
\begin{compactitem}
  \item $Q$ is given by $ (\Atoms(\varphi)\cup \bigcup_{\DBinder.\,\psi\in SubF(\varphi)}\Atoms(\psi))\times\{\downarrow,\uparrow\}$.
  \item $Q_0$  is given by $\{(A,\downarrow)\mid A\in \Atoms(\varphi), \text{ A is initial, and } \{x,\varphi\}\subseteq A\}$.
  \item For all $\psi\in \{\varphi\}\cup \{\theta \mid \DBinder.\,\theta\in SubF(\varphi)\}$, $A\in \Atoms(\psi)$, $dir\in \{\downarrow,\uparrow\}$, and
  $\sigma\in 2^{\Prop}$, $\delta((A,dir),\sigma)$ is defined as follows. If $\sigma\neq A \cap \Prop$, then $\rho((a,dir),\sigma)=\False$. Otherwise, $\delta((A,dir),\sigma)= \xi(dir,A,\psi) \wedge \xi_{ful}(A)$, where:
   \begin{itemize}
     \item $\xi(\downarrow,A,\psi)= \displaystyle{\bigvee_{A'\in \Succ_\psi(A):\, x\notin A'}} (\downarrow,(A',\downarrow))$.

     \item  $\xi(\uparrow,A,\psi)= \left\{\begin{array}{ll}
                  (\uparrow,A)             & \text{ if } A  \text{ is initial }\\
                 \displaystyle{\bigvee_{A'\in \Atoms(\psi):\, A\in Succ_\psi(A')\text{ and } x\notin A'}} (\uparrow,(A',\uparrow))            &  \text{ otherwise}
                \end{array}\right.$

     \item $\xi_{ful}(A) = \displaystyle{\bigwedge_{\DBinder.\,\theta\in A} \,\,\,\bigvee_{B\in \Atoms(\theta):\,   \{x,\theta\}\subseteq B}}(\xi(\downarrow,B,\theta)\wedge \xi(\uparrow,B,\theta))$.
   \end{itemize}

  \item $F_-$ is the set of states whose first component is an initial atom.
  \item For each $\Eventually\psi\in SubF(\varphi)\cup SubF(\widetilde{\varphi})$, $\Frag$ contains the B\"{u}chi component
  consisting of the states $(A,\downarrow)$ such that either $\psi\in A$ or $\Eventually\psi\notin A$.
\end{compactitem}
Correctness of the construction easily follows from Theorem~\ref{theo:characterizationCHLOne}. 
\end{proof}
\end{minipage}
\newpage
\end{changemargin}

%% file: ConfigurationCode.tex
{
\begin{center}
\begin{tikzpicture}[scale=0.97]

\coordinate [label=center:{\footnotesize  $\{2\}\,\cdot$}] (Label2Left) at (0.0,-0.3);
\coordinate [label=center:{\footnotesize  $\emptyset^{\,\nu(2)+1}$}] (Value2Left) at (1.0,-0.27);
\coordinate [label=center:{\footnotesize  $\cdot\,\{1\}\,\cdot$}] (Label1Left) at (2.0,-0.3);
\coordinate [label=center:{\footnotesize  $\emptyset^{\,\nu(1)+1}$}] (Value1Left) at (3.0,-0.27);

\coordinate [label=center:{\footnotesize  $\,\,\cdot\,\,\,\,\{\delta,\#\}\,\,\,\,\cdot\,\,$}] (LabelTrans) at (4.5,-0.3);
\coordinate [label=center:{\footnotesize  $\emptyset^{\,\nu(1)+1}$}] (Value1Right) at (6.0,-0.27);
\coordinate [label=center:{\footnotesize  $\cdot\,\{1\}$}] (Label1Right) at (7.0,-0.3);
\coordinate [label=center:{\footnotesize  $\cdot\,\emptyset^{\,\nu(2)+1}$}] (Value2Right) at (8.0,-0.27);
\coordinate [label=center:{\footnotesize  $\,\cdot \, \{2\}$}] (Label2Right) at (9.0,-0.3);

\coordinate  (OneLeft) at (-0.3,-0.8);
\coordinate  (TwoLeft) at (3.8,-0.8);

\coordinate  (OneRight) at (5.2,-0.8);
\coordinate  (TwoRight) at (9.3,-0.8);
\path[<->, thin,black] (OneLeft) edge node[below] {\footnotesize Left counter-code} (TwoLeft);
\path[<->, thin,black] (OneRight) edge node[below] {\footnotesize Right counter-code} (TwoRight);

\end{tikzpicture}
\end{center}
}

%% file: AdjacentConfigurationCodes.tex
{
\begin{center}
\vspace{-0.2cm}
\begin{tikzpicture}[scale=0.97]

\coordinate [label=center:{\footnotesize  $\{\delta,\#\}\,\cdot$}] (TransLeft) at (0.0,-0.3);
\coordinate [label=center:{\footnotesize  $\emptyset^{\,m+1}$}] (Value1Left) at (1.0,-0.27);
\coordinate [label=center:{\footnotesize  $\cdot\,\,\{1\}\,\,\cdot$}] (Label1Left) at (2.0,-0.3);
\coordinate [label=center:{\footnotesize  $\emptyset^{\,n+1}$}] (Value2Left) at (3.0,-0.27);
\coordinate [label=center:{\footnotesize  $\cdot\,\, \{2\}$}] (Label2Left) at (3.8,-0.3);
\path[<->, thin,black] (1.9,0.3) edge node[above] {\footnotesize $\rho_R$} (4.0,0.3);

 \coordinate [label=center:{\footnotesize  $\,\,\cdot\,\,\,\,\{ \$\}\,\,\,\,\cdot\,\,$}] (LabelMiddle) at (4.8,-0.3);
 \coordinate [label=center:{\footnotesize  $ \{2\}\,\,\cdot$}] (Label2Right) at (5.8,-0.3);
 \coordinate [label=center:{\footnotesize  $\emptyset^{\,n'+1}$}] (Value2Right) at (6.8,-0.27);
\coordinate [label=center:{\footnotesize  $\cdot\,\,\{1\}\,\,\cdot$}] (Label1Right) at (7.8,-0.3);
\coordinate [label=center:{\footnotesize  $\emptyset^{\,m'+1}$}] (Value1Right) at (8.8,-0.27);
\coordinate [label=center:{\footnotesize  $\,\cdot\,\,\{\delta',\#\}$}] (TransRight) at (9.8,-0.3);
\path[<->, thin,black] (5.6,0.3) edge node[above] {\footnotesize $\rho'_L$} (8.0,0.3);

\coordinate  (OneLeft) at (-0.3,-0.8);
\coordinate  (TwoLeft) at (4.0,-0.8);

\coordinate  (OneRight) at (5.6,-0.8);
\coordinate  (TwoRight) at (10.3,-0.8);
\path[<->, thin,black] (OneLeft) edge node[below] {\footnotesize $w_R$: right part of $\rho_C$} (TwoLeft);
\path[<->, thin,black] (OneRight) edge node[below] {\footnotesize $w'_L$: left part of $\rho_{C'}$} (TwoRight);

\end{tikzpicture}
\vspace{-0.2cm}
\end{center}
} 